\documentclass[11pt]{article}
\usepackage[utf8]{inputenc}

\usepackage{authblk}
\usepackage{amsfonts}
\usepackage{amsmath, amssymb}
\usepackage{amsthm}
\usepackage{graphicx}
\usepackage{url}
\usepackage[left=2.5cm,right=2.5cm]{geometry}

\usepackage{dsfont,enumitem}

\usepackage{xfrac}    
\usepackage{algpseudocode} 
\usepackage[linesnumbered]{algorithm2e}
\RestyleAlgo{ruled} 
\SetKwComment{Comment}{/* }{ */}
\usepackage{caption}

\usepackage{hyperref}
\usepackage[T1]{fontenc}
\usepackage{tikz}
\usetikzlibrary{matrix,calc,shapes}
\usetikzlibrary{arrows.meta}
\usetikzlibrary{positioning}

\newcommand{\orig}{\rho}
\newcommand{\weight}{v}
\newcommand{\cond}{\mathrm{cond}}
\newcommand{\conflictID}{\lab}
\newcommand{\unlock}{\mathrm{unlock}}
\newcommand{\constraint}{\mathrm{constr}}

\newcommand{\state}{\mathrm{state}}
\newcommand{\branch}{\mathrm{branch}}
\newcommand{\true}{1}
\newcommand{\false}{0}
\newcommand{\inputs}{\mathrm{in}}
\newcommand{\outputs}{\mathrm{out}}
\newcommand{\Invar}{\mathbf{Invar}}
\newcommand{\AW}{\mathbf{w}}

\newcommand{\lab}{\mathrm{label}}

\newcommand{\C}{\mathcal{C}}

\newcommand{\labspace}{\mathcal{Y}}
\newcommand{\branchset}{\mathcal{B}}
\newcommand{\conflictset}{\mathcal{C}}
\newcommand{\ledger}{\mathcal{L}}
\newcommand{\LedgerDAG}{D_{\ledger}}
\newcommand{\ConflictDAG}{D_{\conflictset}}
\newcommand{\ConflictGraph}{G_{\conflictset}}
\newcommand{\BranchDAG}{D_{\branchset}}

\let\phi=\varphi
\newcommand{\cone}[3]{\mathrm{cone}_{#2}^{(#1)}\left(#3\right)}
\newcommand{\parent}[2]{\mathrm{par}_{#1}\left(#2\right)}
\newcommand{\child}[2]{\mathrm{child}_{#1}\left(#2\right)}
\newcommand{\maximal}[2]{\max_{#1}\left(#2\right)}
\newcommand{\minimal}[2]{\min_{#1}\left(#2\right)}
\newcommand{\conflictsInCone}[3]{\mathrm{conf}^{(#1)}_{#2}\left(#3\right)}

\newcommand{\hash}{\mathop{\mathrm{hash}}}

\DeclareMathOperator*{\argmax}{arg\,max}

\usepackage{graphicx}
\theoremstyle{plain}
\newtheorem{theorem}{Theorem}

\newtheorem{lemma}{Lemma}
\newtheorem{proposition}{Proposition}
\newtheorem{remark}{Remark}
\theoremstyle{definition}
\newtheorem{definition}{Definition}
\newtheorem{example}{Example}
\newtheorem{ass}{Assumption}

\numberwithin{definition}{section}
\numberwithin{lemma}{section}
\numberwithin{ass}{section}
\numberwithin{remark}{section}
\numberwithin{example}{section}
\numberwithin{proposition}{section}
\numberwithin{theorem}{section}

\usepackage{xcolor}
\newif\ifcomment
\commenttrue

\begin{document}

\title{Reality-based UTXO Ledger}

\author[1]{Sebastian M\"uller}
\affil[1]{\footnotesize{Aix Marseille Universit\'e, CNRS, Centrale Marseille, I2M - UMR 7373, 13453 Marseille, France, sebastian.muller@univ-amu.fr}}

\author[2]{Andreas Penzkofer}
\author[2]{Nikita Polyanskii}
\author[2]{Jonas Theis}
\author[2]{William Sanders}
\author[2]{Hans Moog}
\affil[2]{IOTA Foundation,
10405 Berlin, Germany,research@iota.org}

\maketitle

\begin{abstract}
The Unspent Transaction Output (UTXO) model is commonly used in the  field of Distributed Ledger Technology (DLT) to transfer value between participants. One of its advantages is that it allows parallel processing of transactions, as independent transactions can be added in any order. This property of order invariance and parallelisability has potential benefits in terms of scalability. However, since the UTXO Ledger is an append-only data structure, this advantage is compromised through the presence of conflicting transactions. We propose an extended UTXO Ledger model that optimistically updates the ledger and keeps track of the dependencies of the possible conflicts. In the presence of a conflict resolution mechanism, we propose a method to reduce the extended ledger back to a consistent UTXO Ledger. 
\end{abstract}

\section{Introduction}

The Unspent Transaction Output (UTXO) model is a design common to many cryptocurrencies, including Bitcoin \cite{nakamoto2008bitcoin} and many of its derivatives, Cardano \cite{cardanoEUTXO}, 
and IOTA \cite{2020coordicide}.  In the UTXO model, transactions specify the outputs of previous transactions as inputs and create new outputs spending the inputs. Thus, a transaction consists of a list of inputs and a list of outputs. 
The outputs are associated to users` addresses by certain unlock conditions; in general, an account ``possesses'' a  private key and addresses that allow to spend and receive UTXOs. Accounts then track their balance by maintaining a list of the received (unspent) outputs.
This model differs from the account-based model used in most smart-contract-based cryptocurrencies, for example, Ethereum \cite{buterin2013ethereum}. The latter represents assets as balances within accounts, and transactions describe how these balances change, see Figure \ref{fig:AccountModel}. A comparison of blockchain data-models can be for example found in \cite{DataTypeInBlockchain} and we refer to \cite{vademecum} for a detailed discussion on these two models.

\begin{figure}
    \centering
    \includegraphics[width=0.48\textwidth]{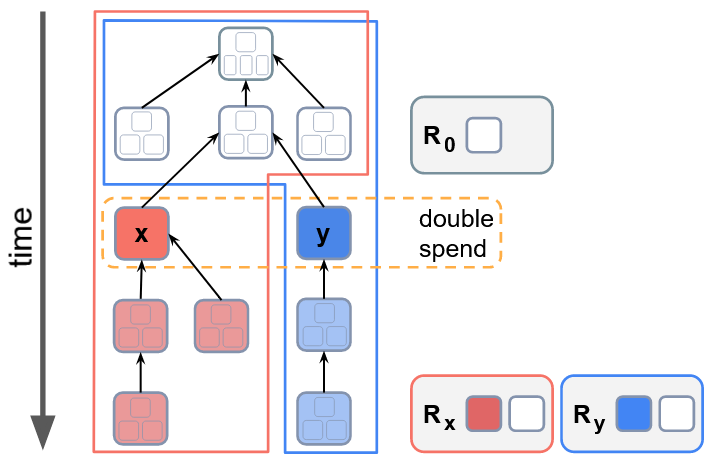}
    \caption{Reality-based Ledger: the double-spending transactions $x$ and $y$ create two realities, $R_x$ and $R_y$. Each reality yields a valid ledger state.}
    \label{fig:realityBasedLedger}
\end{figure}

A conceptual difference is that the account-based model updates user balances globally, while the UTXO model only records transaction receipts. This construction allows transactions to be processed in parallel, leading to performance benefits and possible scalability advantages. A notable difference is that account-based models need a total ordering to resolve conflicts, while in the UTXO model, a total ordering is not necessary. There is, however, a limiting effect on scalability in the presence of conflicting transactions. Here, conflicts are transactions that spent the same output. As the UTXO Ledger is an append-only structure, the conflicts must be sorted out before changes are made. This is currently done with the help of consensus protocols that are based on a (unique) leader or other mechanisms to create a total ordering of the transactions.

We propose a solution to this ``bottleneck'' problem that allows leaderless conflict resolution that does not require total ordering. To this end we define  an augmented data structure, the \emph{Reality-based Ledger}, which processes transactions optimistically and manages possible conflicts until they are resolved. 
As a consequence, we can update the ledger on the arrival of the transactions. This contrasts with a blockchain system where transactions can only be processed once they are included in a batch or block. In further work,~\cite{OTV}, this feature is used to design a stream process-oriented DLT.

\subsection{Results}
A UTXO Ledger induces a partial order on its transactions and thus can be seen as a partially ordered set (poset). Posets are in a one-to-one correspondence\footnote{The correspondence is in fact not one-to-one in the strict sense, but is in a weaker sense involving equivalence relations, see Section~\ref{sec: graph-theoretical notions}.} to directed acyclic graphs (DAGs). Typically, conflicts are excluded from an accepted state in such a ledger, however, this requires that participants reach a consensus on which transactions to add to the ledger. This selection is usually performed by choosing a ``leader'' among the participants, and only this leader can add transactions to the ledger. These transactions are added in batches that are called \emph{blocks}. This leader is a ``centralized'' bottleneck that hinders scalability and annihilates most of the scalability advantages of the UTXO model. Some of these limitations remain even for leaderless consensus protocols that are based on total ordering, as the total ordering requires a ``complete'' or ``non-sharded'' view of the set of transactions.

We define the \textit{Reality-based Ledger} as an augmented data structure of the conflict-free  UTXO Ledger, which may contain conflicting transactions. This Reality-based Ledger can contain many different \emph{realities}, where each of the realities corresponds to a conflict-free  UTXO Ledger, see Figure~\ref{fig:realityBasedLedger} for an illustration. We show that this data structure does not depend on the order of the incoming transactions and forms therefore an eventual consistent distributed data structure. We propose an additional data structure, called the Branch DAG, to manage the dependencies of the different realities. We provide several algorithms for the efficient management and update of the above graph data structures. 
Notably, these algorithms update the data structures on the arrival of new data, thus minimizing the delay between the issuance of transactions and addition to the distributed ledger and enabling a higher degree of stream processing in the underlying DLT.

To obtain the most recent valid ledger state, e.g., to determine the current balance for an account owner, we provide algorithms that extract a reality, based on a weight function that is imposed on the branches. 
Finally, we prove that the augmented data structure can be pruned back into a conflict-free UTXO Ledger in the presence of a conflict resolution mechanism.

\subsection{Related Work}

The benefits and drawbacks of the UTXO model have been discussed extensively, and several extensions of the UTXO model have been proposed. We refer to \cite{vademecum} for an excellent overview. However, we want to note that most of the variants concern either the extension of the UTXO model itself, e.g.,~\cite{cardanoEUTXO}, are proposed to increase their applicability towards smart contracts, or are designed to combine the UTXO and the account-based  model. 
In this paper, we address an extension of the UTXO model to track dependencies of the outputs to allow parallel processing of the transactions.

To allow such a parallel processing of transactions the first step is to move from the traditional blockchain structure of the ledger to a more general DAG-based structure. There are several approaches to improve performance by circumventing the linear chain structure.  Most of them have in common that blocks can reference not only one previous block but also more than one, changing the underlying data structure from a chain to a directed acyclic graph (DAG). This natural idea of using DAGs has become quite popular in the last decade and led to higher throughput and in some cases to similar confirmation latency,
e.g. \cite{Sompolinsky2013AcceleratingBT, NXT,DagCoin, lewenberg2015inclusive,popov2015, SPECTRE, PHANTOM, Prism, Gagol2019, Li2020GHASTBC, OHIE, Narwhal22, Bullshark, CordialMiners, DAGKNIGHT} and the survey paper \cite{Wang2020SoKDI}.

However, parallel writing is only a necessary requirement for efficient parallel processing of transactions but does not yet address the question of  conflict resolution and eventual execution of the transactions. The various DAG protocols, indeed, differ significantly in how conflicts are resolved and transactions are executed. Let us consider these two aspects separately.

\subsubsection{Conflict resolution}
The various DAG protocols differ significantly in the form of how consensus is achieved. In particular, the utilisation of a  DAG data structure does enable but does not require circumventing a total ordering of transactions. For example, in \cite{popov2015} nodes follow and attach to the heaviest DAG, while in most other proposed protocols, e.g. \cite{PHANTOM,lewenberg2015inclusive,Gagol2019, Li2020GHASTBC, OHIE, Prism, DAGRider, Bullshark} consensus is still achieved by constructing a total ordering over the set of transactions. A popular approach that uses the ordering of transactions to achieve consensus is via atomic broadcast protocols. Such protocols allow the network participants to reach a consensus on a (total) ordering of the received transactions, and this linearised output forms then the ledger, e.g., see \cite{HoneyBadger, BEAT}. Improvements of these broadcast protocols are proposed, for example, in  Hashgraph \cite{Hashgraph} and Aleph \cite{Gagol2019} and more recently in Narwhal \cite{Narwhal22} based on the encoding of the ``communication history'' in the form of a DAG. 
The protocols in question serve to alleviate the data dissemination bottleneck of the traditional Nakamoto consensus by decoupling the data dissemination process from the consensus determination procedure. Notable advancements have been made in the realm of consensus determination on top of DAG-based memory pools, as demonstrated in the works of DAG Rider \cite{DAGRider} and Bullshark \cite{Bullshark}. A more comprehensive and abstract examination of these protocols can be found in \cite{BlockDAG}, which provides a description from a broader perspective. There are common points with the approach of IOTA 2.0~\cite{OTV}. A DAG structure serves as a ``testimony'' of the communication among the nodes, and new blocks are used for (implicit) voting on previous blocks. This block structure ``maps down'' to a dependency structure of the contained transactions and UTXOs and the dependencies of the UTXO and the conflict resolution are covered by our model.

{We also want to mention a prominent approach proposed by Prism \cite{Prism}. This methodology outlines the explicit differentiation of the functions of blocks into three distinct categories: proposer blocks, transaction blocks, and voter blocks. The separation of transaction blocks enables participants to initiate transactions, eliminating the requirement for a memory pool. The three categories of blocks assemble into a structured Directed Acyclic Graph (DAG) that facilitates an efficient means of voting on ``leader blocks,'' resulting in consensus through total ordering. Again, the dependency structure and conflict resolution, in the setting of UTXO model, can be described with the model in our paper. }

Total-ordering is, however, not necessary to achieve consensus. This approach is pursued by \cite{OTV}. Similar to the DAG-based memory pools a causal dependency of the blocks is used to determine confirmation of the contained transactions. This approach, however, requires an active tracking of the dependencies and a certain ``confirmation weight''. Our results on tracking of conflicts and their resolution is an important ingredient for the protocol proposed in \cite{OTV}, but the concepts are natural for an optimistic execution and \textit{a posteriori} conflict resolution.

\subsubsection{Execution and view change} 
Parallel booking and execution of transactions is essential for high throughput and scalability. There are three existing proposals for achieving this, each of which is contingent upon the method of consensus attainment. In the case of total-ordering, parallel execution of transactions following consensus can be accomplished. For instance, the implementation of Prism \cite{Prism} in \cite{Prism10000}  employed a scoreboard technique to parallelize the execution of ordered UTXO transactions. It was demonstrated in \cite{prismSC} that Prism can support smart contract platforms, with the execution of smart contracts rather than consensus serving as the bottleneck in their implementation. 

Another possibility building on a total-ordering is the optimistic booking of the transactions and a \textit{a posteriori} conflict resolution via total ordering. For instance, this is the direction of a current Leios proposal in Cardano \cite{OuroLeios} and in Sui \cite{suiWP}.

{Our contribution is not limited to merely parallel booking but encompasses the comprehensive mechanism for processing new transactions and instantaneously updating the ledger state as perceived by the majority of network participants. Our approach actively constructs a DAG, referred to as the Ledger DAG, which encodes the dependencies between transactions. This DAG is generated prior to consensus and facilitates the tracking of dependencies between pending or conflicting transactions. This natural idea has been explored in various academic papers. It does rely on the construction of a dependency graph that takes the form of a DAG and encodes the causal dependencies of the transactions, e.g. \cite{ConcurrencyToSC, OptSmart}, or on a pre-ordering, \cite{BlockSTM}. These works focus on the execution of smart contracts in an account-based model, i.e. with no local states, while our work covers the situations of local state transactions as in the UTXO-model.
Finally, let us note that our approach does not rely on total-ordering; however, it still supports it, thus eliminating the dependence on the linear structure of total-ordering for view changes and transaction confirmation.}

In conclusion, it is important to highlight that the preceding discussion does not aim to deliver a comprehensive depiction of the current state of all Directed Acyclic Graph (DAG)-based protocols and their multifaceted designs. The intention is not to exhaustively cover every aspect of these protocols, but rather to offer an overview of the  aspects that relate closely to our work. We encourage the readers to delve into the referenced works for a more detailed understanding and to explore further literature for additional perspectives and developments that may not be included in this section.

\subsubsection{Beyond DLTs}

From  a general point of view, our approach is natural and was already been successfully applied in various settings. It relies on constructing  systems that allow working effectively and efficiently with inconsistent information and where ``context switching'' is a low-cost operation; e.g.~\cite{TMS}. We also want to note that this approach can be found in the field of belief revision. 

Finally, we want to draw some connections with the field of replicated invariant data types. The UTXO model, if distributed on multiple nodes, falls into the class of a replicated data types as the nodes can make concurrently changes to the data. Its append-only design and the invariant constraints render the operation of adding a transaction not commutable. This means that the order of the incoming transaction matters for the construction of the ledger. However, if there is a deterministic rule on how to process conflicts or a ``leader'' that pre-filters the transactions, this data structure can be turned into a conflict-free replicated data type (CRDT), e.g.,~see~\cite{CRDT}. The proposed Reality-based Ledger is already a CRDT without the need of a consensus mechanism since all operations on the involved data types are commutative.

In the context of replicated data structure, some works increase the performances in ``augmenting'' the data structure and distinguishing between different notions of consistency, e.g.,~\cite{redblue, automatingChoiceOfConsitency}. From a conceptual idea, this resembles our approach. However, the concrete proposals are distinct due to different assumptions on the communication model and field of applications. To our knowledge, our method is the first that allows a leaderless ``ex-post'' conflict resolution on these types of models. 

\subsection{Structure of the Paper}

The document is structured as follows. In Section \ref{sec: Standard UTXO} we provide an introduction to a standard conflict-free UTXO Ledger. In Section \ref{sec: graph-theoretical notions} we give an overview of some of the graph theoretical preliminaries used in this paper. In Section \ref{sec: Reality-based Ledger} we introduce the concept of a Reality-based Ledger and additional data structures that provide the necessary tools to manage this novel type of ledger. In Section \ref{sec: operations on the ledger} some of the core operations to maintain and access the ledger are presented.

We employ several graph structures to efficiently manage the Reality-based Ledger. Table~\ref{tab: graph overview} gives an overview of the utilised graphs. 

\begin{table}[b]
    \centering
    \begin{tabular}{c|c|c}
         Graph &  Vertices & Edges   \\   \hline\hline
         UTXO DAG  & in-, outputs, transaction IDs & spending relations \\
          Ledger DAG & transactions  & spending relations \\
         Conflict DAG & conflicts  &  conflict dependencies \\
         Conflict Graph & conflicts  &  conflict relations  \\
        Branch DAG  & branches  &  branch dependencies 
    \end{tabular}
    \caption{Overview of the graphs used in this paper.}
    \label{tab: graph overview}
\end{table}

\begin{figure}
    \centering
    \includegraphics[width=0.3\textwidth]{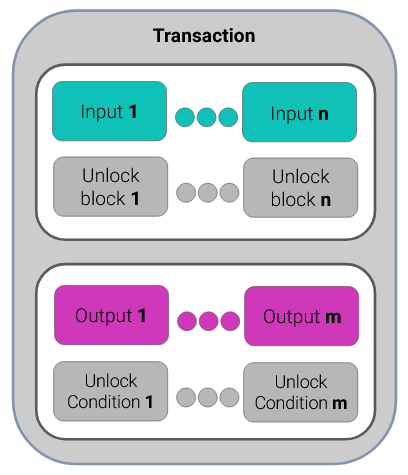}
    \caption{Simplified transaction layout. The fund owner signs the transaction in the unlock blocks. Inputs are consumed and new outputs are created. The new outputs can be spent once the unlock conditions are satisfied.}
    \label{fig:transactionLayout}
\end{figure}

\begin{figure*}[t]
    \centering
    \includegraphics[width=0.8\textwidth]{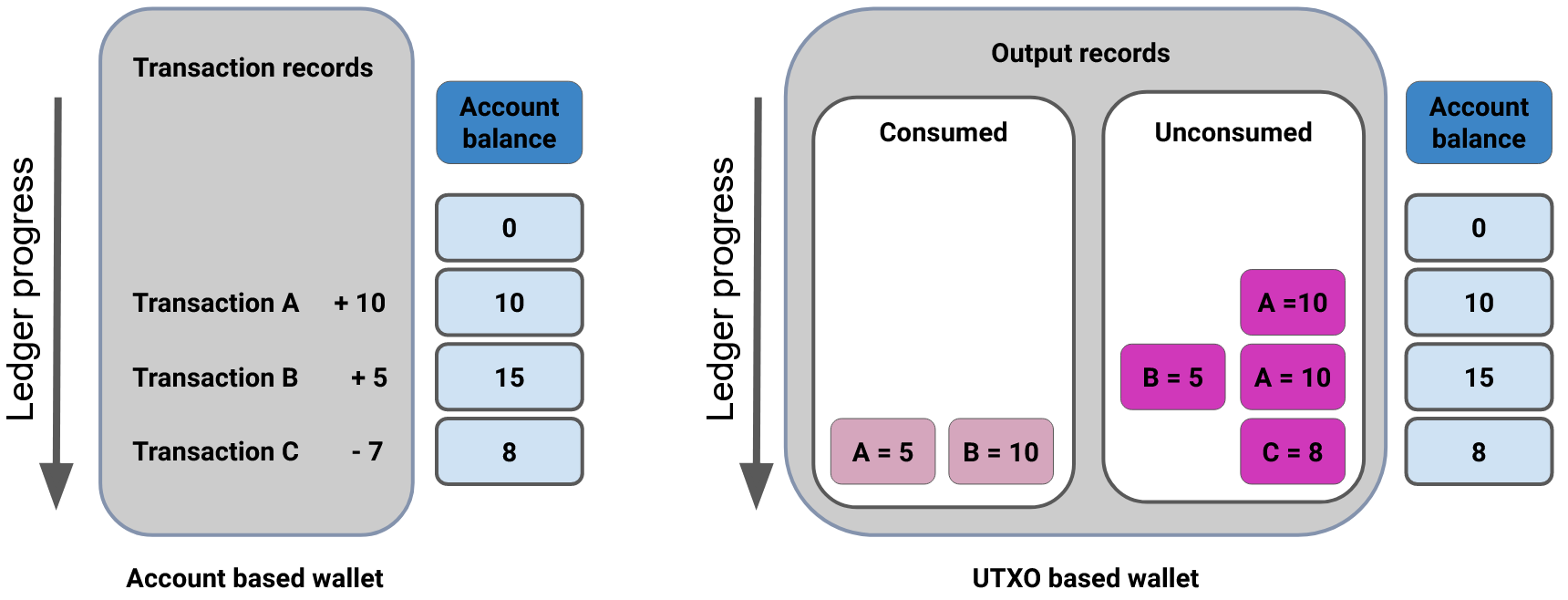}
    \caption{Comparison between account- and UTXO-based model. }
    \label{fig:AccountModel}
\end{figure*}

\section{Conflict-Free UTXO Ledger}\label{sec: Standard UTXO}

In the standard UTXO model, transactions specify the outputs of previous transactions as inputs and create new outputs spending (or consuming) the inputs. 
Thus, a transaction consists of a list of inputs and a list of unique outputs. To every output, we associate a unique reference or output ID. Typically such an output ID is created with the involvement of a hash function.
\begin{remark}\label{rem: hash function}
A collision-resistant hash function is used to map data of arbitrary size to a fixed-size binary sequence, i.e., $\hash: \{0,1\}^*\to \{0,1\}^h$. Moreover, it is required that it is practicably impossible to find for a given sequence $x$ another sequence $x'$ such that $\hash(x) =\hash(x')$.
Throughout the remainder of the paper, we assume that a particular hash function is fixed and used by all participants.
\end{remark}
For example, the output ID could be created through the concatenation of the index of the output within the transaction and the hash of the transaction's content.
Every output represents a specific amount of the underlying cryptocurrency. The value of all inputs, i.e., spent outputs, must equal the value of all outputs of a transaction.
Every output can be spent only once and, hence, value is conserved overall.
With each output comes an unlock condition, which declares by whom and under which conditions it can be spent. With each input comes an unlock block containing a proof that the transaction issuer is allowed to spend the inputs and fulfills the unlock condition, e.g., a signature proving ownership of a given input's address. We refer to Figure~\ref{fig:transactionLayout} for a general transaction layout. In Section~\ref{sec:realities} we propose a more general description of this model. 

In this model, an account that is controlled by an entity holding the corresponding private / public key pair, is a collection of UTXOs that can be unlocked through the key pair. This type of book keeping of balances and transactions differs fundamentally from the account-based model, such as it is used in Ethereum \cite{buterin2013ethereum}. In the account-based model funds are represented as balances within accounts and transactions describe how these balances change, see Figure \ref{fig:AccountModel}. 

Let us define the UTXO Ledger model more formally. We follow the approach of~\cite{IOTASC}. Note that for the purpose of this work we consider a simplification of the UTXO models used in practice.
\begin{definition}[Output and input]
An \textit{output} is a pair of a value $\weight\in \mathbb{R}^+$ and an unlock condition $\cond$. We write $o=(\weight, \cond)$ to denote the output. An \textit{input} $i$ is a reference to an output. In such a case, we say the input consumes the output. 
\end{definition}

\begin{definition}[Transaction]\label{def: transaction}
A transaction $x$ is a collection of inputs $\inputs(x),$ outputs $\outputs(x)$, and an unlock data $\unlock(x)$:
\begin{enumerate}[leftmargin=*]
    \item $\inputs(x)=(i_1,\ldots, i_n)$ is a list of inputs, i.e., references to outputs. We say that those outputs are \textit{consumed} or \textit{spent} by transaction $x$.
    \item $\outputs(x)=(o_1,\ldots, o_m)$ is a list of new outputs produced by transaction $x$.
    \item $\unlock(x)$ is a data which unlocks the inputs. This is usually done by cryptographic proof of authorization that ensures that the issuer of the transaction satisfies the condition $\cond$ of the consumed outputs.
\end{enumerate}
\end{definition}

\begin{definition}[Ledger]\label{def: ledger}
The \emph{ledger} is a set of transactions and denoted as $\ledger$. 
\end{definition}

\begin{definition}[Ledger state]\label{def: ledger state}
The \emph{ledger state}, written as $\state(\ledger)$, is the set of all outputs that are not consumed by a transaction in the ledger $\ledger$. In other words, the ledger state is the set of outputs, for which no input exists that references them.
\end{definition}

The ledger progresses through the addition of new transactions. Furthermore, it is an append-only structure, i.e., transactions can only be added and not removed from the ledger.

In a distributed system, the append-only nature of this data structure makes it necessary that the operators of the distributed ledger have consensus on which transactions should be added to the ledger. In a blockchain setting, such as Bitcoin, this can be achieved by the selection of a leader who typically extends the longest chain. The longest chain then determines which transactions  are included in the ledger and in which order. 

To provide consistency there can be specific \emph{ledger constraints} which ought to be fulfilled before a certain  transaction $x$ can be added to the ledger $\ledger$. The ledger constraints are generally enforced by transaction validation rules applied prior to addition of the transaction to the ledger. We define $\constraint(\ledger,x)=\true$ if all imposed constraints are satisfied by transaction $x$ and $\constraint(\ledger,x)=\false$ otherwise. 
The constraints of a transaction being added to the ledger adhere typically to the following assumption:
\begin{ass}[Ledger constraints]\label{ass:constraints}
A transaction $x$ is added to the ledger $\ledger$ if it follows the following rules:
\begin{enumerate}[leftmargin=*]
    \item the transaction $x$ is syntactically correct;
    \item the sum of values of $\inputs(x)$ equals the sum of values of $\outputs(x)$;
    \item the unlocking data $\unlock(x)$ is valid;
    \item $\inputs(x)$ are references to existing unspent (not yet spent) outputs in the ledger state  $\state(\ledger)$;
\end{enumerate}
\end{ass}

\begin{definition}[Consistent Ledger]\label{def: consistent ledger}
We say that the ledger  $\ledger$ is \emph{consistent} if and only if $\constraint(\ledger,x)=\true$ for all $x\in \ledger$.
\end{definition} 

The UTXO Ledger starts at the so-called \emph{genesis} $\orig$, i.e., the transaction that is the ultimative predecessor of any transaction of the UTXO Ledger. The genesis transaction only contains outputs and no inputs. 
Each new transaction is an atomic update of the ledger and the ledger state. A transaction $x$ is added to a consistent ledger $\ledger$ if $\ledger\cup \{x\}$ is a consistent ledger or, equivalently, if $\constraint(\ledger, x)=\true.$
The above constraints or consistency rules imply that each output can be consumed by at most one transaction and thus a consistent ledger can not contain a so-called \textit{double spend}. 

An important property of the UTXO Ledger is that the validity of the state update (adding a new transaction) can be determined  by only using the context of the transaction itself, i.e., inputs, outputs, and unlock conditions. This allows a certain degree of parallelism and turns the UTXO Ledger into a partially ordered data structure.

\begin{remark}\label{rem: invariant ledger}
We can talk about a ledger invariant $\Invar(\ledger)$ which is preserved by each addition of a transaction to the ledger. In other words, $\Invar(\ledger)=\Invar(\ledger')$ for any two consistent ledgers $\ledger$ and $\ledger'$.  The prime example is that the sum of the values of all unspent outputs in the ledger state remains constant.
\end{remark}

\section{Graph Theoretical Preliminaries}\label{sec: graph-theoretical notions}

In this section, we summarize basic graph theoretical notations and results that are used in the remaining part of the paper.

The set of integers between $1$ and $m$ is denoted by $[m]$. A \emph{graph} ${G}$ is a pair $(V,E)$, where $V$ denotes the set of vertices and $E$ denotes the set of edges. A graph is called \emph{directed} if every edge has its direction, e.g.,  for an edge $(u,v)$, the direction goes from $u$ to $v$. 

\begin{definition}[DAG]
A \emph{directed acyclic graph (DAG)} is a directed graph with no directed cycles, i.e.,  by following the directions of edges, we never form a closed loop.
\end{definition}

A vertex $v$ in a graph ${G}=(V,E)$ is called \textit{adjacent} to a vertex $u$ if $(u,v)\in E$. An edge $e\in E$ is said to be adjacent to a vertex $v\in V$ if $e$ contains $v$. The \textit{out-degree} and \textit{in-degree} of a vertex $v$ in a directed graph $G=(V,E)$ is the number of adjacent edges of the form $(v,u)$ and, respectively, $(u,v)$. A vertex in a graph is called \textit{isolated} if there is no edge adjacent to it.
\begin{definition}[Neighbours in a graph]~\label{def: neighbours in graph}
Let ${G}=(V,E)$ be a graph. For a vertex $v\in V$, define the \emph{set of neighbours} (or ${G}$-neighbours), written as $N_{{V}}(v)$\footnote{In the remainder of the paper, we will often identify the graph with its vertex set, since for a given set of vertices $V$, we will have only one DAG $D=(V,E)$. Thereby, the set of neighbours $N_{{V}}(v)$ and other concepts that use $V$ as a subscript will be clear from the context.}, to be the vertices adjacent to $v$. 
\end{definition}
\begin{definition}[Parents, children and leaves in a DAG]
Let ${D}=(V,E)$ be a DAG. For a vertex $v\in V$, define the set of \emph{parents}, written as $\parent{V}{v}$, to be the set of vertices $u\in V$ such that $(v,u)\in E$. Similarly, we define the set of \emph{children}, written as $\child{V}{v}$, to be the set of vertices $u\in V$ such that $(u,v)\in E$. A vertex $v\in V$ with in-degree zero is called a \emph{leaf}. 
\end{definition}

\begin{definition}[Partial order induced by a DAG]
Let ${D}=(V,E)$ be a DAG. We write $u\le_{{V}} v$  for some $u,v\in V$ if and only if there exists a directed path from $u$ to $v$, i.e., there are some vertices $w_0=u,w_1,\ldots,w_{s-1},w_s=v$ such that  $(w_{i-1},w_{i})\in E$ for all $i\in[s]$.  Furthermore, we write $u<_{{V}}v$ if $u\le_{{V}} v$ and $u\neq v$.
\end{definition}
Note there could be different DAGs producing the same partial order. The DAG with the fewest number of edges that gives the partial order $\le_{{V}}$ is usually called the \emph{transitive reduction} of ${D}$ or the \emph{Hasse diagram} of $\le_{{V}}$. In the following definition, we give a more general definition of the minimal subDAG of ${D}=(V,E)$ induced by a set of vertices $S\subseteq V$ which coincides with the transitive reduction of ${D}$ when $S=V$. 

\begin{definition}[Minimal subDAG induced by a set of vertices]\label{def: minimal subDAG}
Let ${D}=(V,E)$ be a DAG.  For a subset of vertices $S\subseteq V$, we define the \emph{minimal subDAG} of ${D}$ induced by $S$ to be the DAG ${D'}=(V',E')$ whose vertex set is $V'=S$ and there is an edge $(v,u)\in E'$ if and only if $u,v\in S$, $v<_{{V}}u$ and there is no $w\in S\setminus \{u,v\}$ such that $v<_{{V}} w <_{{V}}u$. 
\end{definition}

\begin{definition}[Maximal and minimal elements]\label{def: min and max elements} Let ${D}=(V,E)$ be a DAG and let $\le_{{V}}$ be the partial order induced by ${D}$. For a subset of vertices $S\subseteq V$,  an element $u\in S$ is called \emph{${D}$-maximal} (\emph{${D}$-minimal}) in $S$ if there is no $v\in S\setminus\{u\}$ such that $u\le_{{V}}  v$ ($v\le_{{V}}  u$). Define $\maximal{V}{S}$ and $\minimal{V}{S}$ to be the set of ${D}$-maximal and, respectively, ${D}$-minimal elements in $S$.
\end{definition}

\begin{remark}
The maximal (minimal) elements of a DAG $D$ are also called the geneses (tips) of $D$. Usually, we consider DAGs with only one genesis, whereas the number of tips can be large. 
\end{remark}

\begin{definition}[Future and past cones] \label{def: future and past cones}
Let ${D}=(V,E)$ be a DAG. For $x\in V$, define the \emph{past cone} of $x$ in ${D}$, written as $\cone{p}{V}{x}$ 
to be the set of all vertices $y\in V$ such that $x\le_{{V}} y$. Similarly, define the \emph{future cone} of $x$ in ${D}$, written as $\cone{f}{V}{x}$ 
to be the set of all vertices $y\in V$ such that $y\le_{{V}} x$. 
\end{definition}

\begin{definition}[Future-closed and past-closed sets]\label{def: future and past closed}
Let ${D}=(V,E)$ be a DAG. A subset $S\subset V$ is called \emph{${D}$-past-closed} if and only if for every $u\in S$, the past cone $\cone{p}{{V}}{u}$ is contained in $S$. 
Similarly, a subset $S\subset V$ is called \emph{${D}$-future-closed} if and only if for every $u\in S$, the future cone $\cone{f}{{V}}{u}$  is contained in $S$. 
\end{definition}

We conclude with a definition of maximal independent sets for general graphs.

\begin{definition}[Maximal Independent Set]\label{def: maximal independent set}
Let ${G}=(V,E)$ be a finite graph. A subset $S\subset V$ is an \emph{independent set} if and only if for every two vertices $u,v \in S$  there is no edge connecting the two, i.e., $(u,v)\notin E$. An independent set $S$ is called a \emph{maximal independent set} if and only if there is no other independent set $S'$ such that  $S\subsetneq S'$.
\end{definition}

\section{Reality-based Ledger}\label{sec: Reality-based Ledger}

In Section \ref{sec: Standard UTXO} we described the model of a conflict-free UTXO Ledger that is suitable for an environment where transactions are pre-filtered by a consensus mechanism. Since that ledger was conflict-free a valid ledger state could be readily extracted, see Definition~\ref{def: ledger state}.

To alleviate the restriction of requiring a conflict-free data structure, we propose an augmented version of the standard conflict-free UTXO Ledger model that allows more than one output spend. The ledger $\ledger$ continues to be defined by Definition~\ref{def: ledger}, however, the total set of transactions must not be conflict-free. We will derive a concept, called a \textit{reality}, which allows to reduce $\ledger$ to a subset of transactions that yield a valid (\textit{Reality-based}) Ledger state, see also Definition~\ref{def: ledger state}. We refer to Figure \ref{fig:DependenciesDefinition} for an overview on the dependencies of the principal definition that are required to describe the reality-based ledger and to Figure \ref{fig:notations} for an overview of the used notations.

\tikzset{%
            base/.style = {rectangle, rounded corners, draw=black,
                           minimum width=4cm, minimum height=1cm,
                           text centered, font=\sffamily},
          decision/.style = {base, diamond, inner sep=-5pt},
  activityStarts/.style = {base, fill=blue!10},
       startstop/.style = {base, fill=red!10},
    activityRuns/.style = {base, fill=green!30},
         processOpt/.style = {base, minimum width=2.5cm, fill=orange!10,
                           font=\ttfamily},
       process/.style = {base, minimum width=2.5cm, fill=orange!25,
                           font=\ttfamily,},
 processExt/.style = {base, minimum width=2.5cm, fill=green!10,
                           font=\ttfamily},                        
}

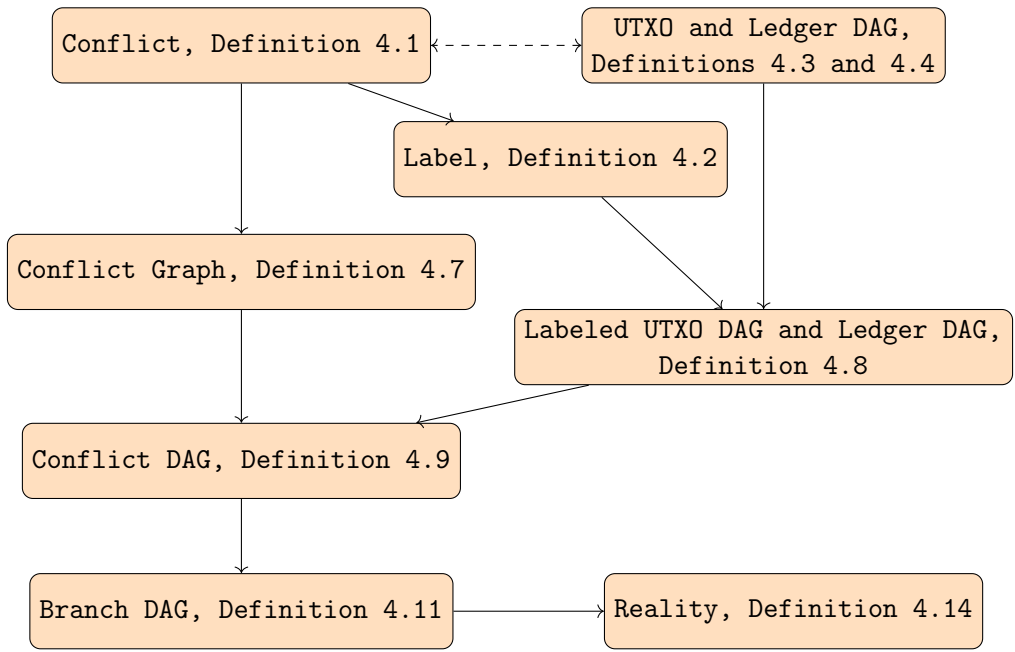
\begin{figure}[ht!]
    \centering
\begin{tikzpicture}[node distance=2cm,
    every node/.style={fill=white, font=\sffamily}, align=center]
  \node (conflicts)     [process]          {Conflict, Definition \ref{def:conflicts}};
 
  \node (UTXODAG)             [process, right =  2cm of conflicts]              {UTXO and Ledger DAG, \\  Definitions \ref{def: UTXO DAG} and  \ref{def: ledgerDAG}};
 
 \node (label) [process, below right = 0.5cm and -0.5cm of conflicts] {Label, Definition \ref{def:label}};
 
 \node (ConflictGraph) [process,  below = 2cm of conflicts ] {Conflict Graph, Definition \ref{def: Conflict Graph}}; 
 
\node (labelledUTXODAG)             [process, below = 3cm of UTXODAG]              {Labeled UTXO DAG and  Ledger DAG, \\ Definition \ref{def:labeled UTXLedgerDAGs}};

 \node (conflictDAG)             [process, below = 1.5cm  of ConflictGraph]              {Conflict DAG, Definition \ref{def: Conflict DAG}};

  \node (branchDAG)             [process, below of = conflictDAG]              {Branch DAG, Definition \ref{def:branchDAG}};
  
  \node (realities)             [process, right  = 2cm of branchDAG]              {Reality, Definition \ref{def: reality}};
   \draw[<->, dashed]           (UTXODAG) --  (conflicts) ;
   \draw[->]           (conflicts) --  (label) ;  
   \draw[->]           (conflicts) --  (ConflictGraph) ;  
   \draw[->]           (label) --  (labelledUTXODAG) ;
   \draw[->]     (UTXODAG)   -- (labelledUTXODAG);
   \draw[->]           (labelledUTXODAG) --  (conflictDAG) ;
    \draw[->]           (ConflictGraph) --  (conflictDAG) ;
    \draw[->]           (conflictDAG) -- (branchDAG) ;
   \draw[->]           (branchDAG) -- (realities);
\end{tikzpicture}
\caption{Dependencies of the main definitions around the reality-based ledger. }
\label{fig:DependenciesDefinition}
\end{figure}

The purpose of the augmentation is to identify and track possible conflicting transactions. We add this information to each transaction and output. For a given transaction we may set an additional flag or \textit{label} $\conflictID$.

Let us highlight that the transaction layout does not change and $\conflictID$ is only assigned a value if needed.
We can therefore think of an output as a triplet of a value $v$, an unlock condition $\cond$, and a  label $\conflictID$, i.e., $o=(v, \cond, \conflictID)$. Let us now define what we mean by conflicts.  

\begin{definition}[Conflicts]\label{def:conflicts}
A transaction $x\in \ledger$ is called a \textit{conflict} if and only if there exists a transaction $y\in \ledger\setminus\{x\}$ such that $x$ and $y$ contain at least one same input. The set of all conflicts is denoted by $\conflictset$ and dubbed the \emph{conflict set} of the ledger $\ledger$.
\end{definition}

{\noindent We can now define how we set the label  $\conflictID$:}
\begin{enumerate}[leftmargin=*]
    \item if a transaction $x$ is not a conflict, the label is not set;
    \item otherwise the label is set to a generic unique reference to the transaction, e.g., the transaction ID. 
\end{enumerate}
{\noindent More formally we define the label as follows.}

\begin{definition}[Label]\label{def:label}
We define $\bot$  to be the label of the  genesis $\orig$. 
Let $\labspace$ be a label space such that $\bot\in\labspace$, and $\lab: \ledger \to \mathcal{Y}$ be a function with the following properties:
\begin{enumerate}[leftmargin=*]
   \item if $x \in \ledger \setminus \{ \C\}$, then $\lab(x)=\bot$;
    \item the restriction of the function $\lab$ on the set $\C\cup \{\orig\}$ is injective, i.e., the image  $\lab(\C\cup \{\orig\}):=\{\lab(x): x\in \C\cup \{\orig\}\}$ has size $|\C|+1$.  
\end{enumerate} 
\end{definition}
\begin{remark}
One natural choice to set a unique (with high probability) label function is to utilize a hash function $\hash(\cdot):  \{0,1\}^*\to \{0,1\}^h$ (cf. Remark~\ref{rem: hash function}) for large enough $h$. Then the label set $\mathcal{Y}$ is $\{0,1\}^h \cup \bot$.
\end{remark}
\begin{remark}
As two conflicting transactions $x$ and $y$ may be perceived at different times, the detection of a conflict can only be achieved after having received both transactions. If transaction $x$ is perceived first, we do not see it yet as a conflict and do not set any flag. Only, when transaction $y$ arrives, we can identify $x$ as a conflict and set both flags for $x$ and $y$. 
\end{remark}

\begin{figure}[t!]
\begin{tabular}{ll}
    \multicolumn{2}{l}{\textbf{Set symbols}}\\
    $\conflictset$ &  set of conflicts \\
    $\ledger$ &  ledger or set of transactions  \\
    \multicolumn{2}{l}{\textbf{DAG-related notation}}\\
    $D$=$(V,E)$  &  directed acyclic graph (DAG) with vertex \\
    & set $V$ and edge set $E$ \\
    $\LedgerDAG$ & Ledger DAG  \\
    $\ConflictDAG$ & Conflict DAG  \\
    $\child{V}{x}$   &  set of children of vertex $x$ in DAG $D$=$(V,E)$ \\
    $\parent{V}{x}$ &  set of parents of vertex $x$ in DAG $D$=$(V,E)$ \\
    $\cone{f}{V}{x}$    & future cone of vertex $x$ in DAG $D$=$(V,E)$ \\
    $\cone{p}{V}{x}$  &  past cone of vertex $x$ in DAG $D$=$(V,E)$ \\
    $\lab^{(p)}(x)$ & set of labels in past cone of $x$ in $\LedgerDAG$\\
    \multicolumn{2}{l}{\textbf{Order and relationship definitions}}\\ 
    $\le_V$ &  partial order on set $V$ (usually induced by\\
    & a given DAG $D$=$(V,E)$) \\
    $N_V(x)$ &  set of neighbours of a vertex $x$ in \\
    & graph $G$=$(V,E)$ \\
    $\maximal{V}{S}$ &  set of maximal elements in set $S$ (maximal \\
    & according to DAG $D$=$(V,E)$) \\
    $\minimal{V}{S}$ &  set of minimal elements in set $S$ (minimal \\
    & according to DAG $D$=$(V,E)$) \\
\end{tabular}
\caption{Overview of the main notations.}
\label{fig:notations}
\end{figure}

In the presence of this label, we can remove the forth constraint from the ledger, see Assumption~\ref{ass:constraints}. However, to avoid the data structure becoming ``meaningless'' we still need a notion of ``consistency'' as we will see in Assumption~\ref{ass:constraints2}, and a mechanism  to track the dependencies of the conflicts.

The next part of this section  defines several data structures that can be derived from the UTXO inter-dependencies. These structures are used to track conflicting transactions without the need for consensus. More precisely, in Section~\ref{sec: Standard UTXODAG} we will explain how the UTXO transactions and their in- and outputs result in a DAG structure. In Section~\ref{sec:conflictGraph} we present how we can use the UTXO data structure to manage the conflicting transactions efficiently. 

In Section~\ref{sec:conflictDAGconflictGraph} the information contained in the UTXO DAG is split into the Conflict Graph, which keeps track of the conflicting transactions, and the Conflict DAG, which describes the inherited dependencies of conflicts. Finally, in Section~\ref{sec:branchDAG} branches are introduced, which form a possible non-conflicting state of the ledger. Combining non-conflicting branches can create maximally independent sets of conflicts called realities. Each  reality can be associated to a consistent ledger, see Theorem \ref{thm:Rledger}.

\subsection{UTXO and Ledger DAGs }\label{sec: Standard UTXODAG}
We introduced the concept of UTXO and defined UTXO based transactions as an operation spending inputs and creating outputs in Section~\ref{sec: Standard UTXO}.
The inputs and outputs in one transaction are ``atomic'' in the sense that either all inputs are consumed and all outputs are created or the transaction is not added to the ledger at all. The atomic nature of a transaction is represented by a unique transaction ID. In our graphical representation, these dependencies are expressed  using an additional vertex identified with the corresponding transaction ID, see Figure~\ref{fig: UTXO DAG}.

The collection of all transactions since genesis, i.e., the ledger $\ledger$,
provides the content for a DAG, which we call the UTXO DAG.
\begin{definition}[UTXO DAG]\label{def: UTXO DAG}
The vertex set of the UTXO DAG  
consists of all in- and outputs and all transaction IDs.  The interrelations between these form the set of directed edges. More specifically, directed edges exist from inputs to outputs, from the transaction ID to its inputs, and from the outputs to the ID of the transaction creating these outputs. We allow here to appear inputs several times as vertices, turning the vertex set formally into a multi-set. This allows to track possible double spends.
\end{definition}
\begin{example}
For a simplified illustration of an example for such a DAG see Figure~\ref{fig: UTXO DAG}. The depicted UTXO DAG contains five transactions.
\end{example} 
\begin{definition}[Ledger DAG]\label{def: ledgerDAG}
We define the Ledger DAG $\LedgerDAG$ to be a DAG whose vertex set is the ledger $\ledger$. There is a directed edge $(x,y)$ in the edge set of  $\LedgerDAG$ if and only if an input of $x$ references an output of $y$.
\end{definition}

\begin{remark}
The results derived in this paper are in respect to the Ledger DAG. However, due to the atomic nature of transactions the results also apply for the UTXO DAG.
\end{remark}

\begin{figure}
    \centering
    \includegraphics[width=0.48\textwidth]{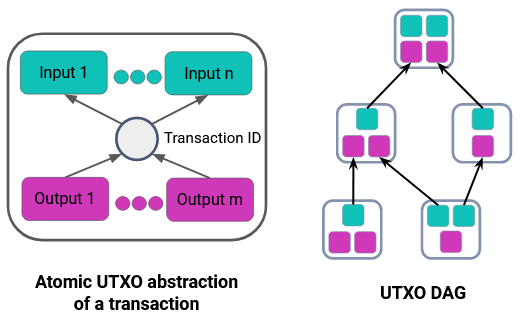}
    \caption{Atomic UTXO abstraction of a transaction and an example of a UTXO DAG. }
    \label{fig: UTXO DAG}
\end{figure}

Equipped with the notion from Section~\ref{sec: graph-theoretical notions},  we write $\le_{\ledger}$ to denote the partial order on the ledger $\ledger$ induced by $\LedgerDAG$. In other words,  $y\le_{\ledger} x$ if transaction $y$ spends (indirectly) from $x$. Note that the genesis $\orig$ is the only $\LedgerDAG$-maximal element in $\ledger$. Further, we write 
$\cone{p}{\ledger}{x}$ to denote the Ledger past cone and  
$\cone{f}{\ledger}{x}$ to denote the Ledger future cone.

As a typical rule in a DLT with a UTXO model, an output can only be spent once. Thus, if there are multiple transactions that attempt to spend the same output, it is the role of the consensus mechanism to select at most one transaction that is allowed to consume the output. 
Once some consensus mechanism decided on which conflicts to keep and which to reject, we can reduce or prune the augmented ledger as described in Section \ref{sec:ConflictPruning}.

\subsection{Conflict Graph}\label{sec:conflictGraph}

In Definition~\ref{def:conflicts} we introduced the notion of conflicts. Due to the causal dependency of ordered transactions, transactions can be conflicting even if they do not consume the same output.

\begin{definition}[Conflicting transactions]\label{def:conflictingTx}
Two distinct transactions $x,y\in \ledger$ are \emph{directly conflicting} if they have at least one input in common. 
Two distinct transactions $x_1,y_1\in \ledger$ are said to be \emph{indirectly} conflicting  if there exist distinct $x_2,y_2\in \ledger$ with either $x_1 <_{\ledger} x_2$ and $y_1\le_{\ledger} y_2$ or $x_1 \le_{\ledger} x_2$ and $y_1<_{\ledger} y_2$ such that $x_2$ and $y_2$ are directly conflicting. Two transactions are said to be \emph{conflicting} if they are directly or indirectly conflicting.
\end{definition}

\begin{remark}\label{rem: conflicting transaction are not conflicts}
Note that due to Definition~\ref{def:conflicts} some (or possibly the majority of) conflicting transactions are not necessarily conflicts. On the other hand, if two transactions are directly conflicting, then they are conflicts.
\end{remark}

\begin{definition}[Conflict-free set and conflicting sets]\label{def: conflicting sets}
A subset of transactions $S\subseteq \ledger$ is called \textit{conflict-free} if it does not contain any two  conflicting transactions. We also say that $S_1\subseteq \ledger$ is \textit{conflict-free with respect to} $S_2\subseteq \ledger$ if there is no $c_1\in S_1$ and $c_2\in S_2$ such that $c_1$ and $c_2$ are conflicting. Alternatively, $S_1$ is \textit{conflicting} with $S_2$ if $S_1$ is not conflict-free with respect to $S_2$.
\end{definition}
By Remark~\ref{rem: conflicting transaction are not conflicts}, conflicting transactions are not necessarily conflicts. However, the $\LedgerDAG$-maximal transactions that are conflicting with a given transaction have to be conflicts as described below. 

\begin{proposition}\label{prop: conflicting transactions}
For a transaction $x\in \ledger$, define $C=C(x)$ to be the set of   transactions that are conflicting with $x$. Then $C$ is $\LedgerDAG$-future-closed and it holds that $\maximal{\ledger}{C}\subseteq \conflictset$.
\end{proposition}

\begin{proof}
By Definition~\ref{def:conflictingTx}, if $x$ is conflicting with a transaction $y\in \ledger$, then $x$ is conflicting with $z$, where $z$ is any transaction $z\in \ledger$ such that $z\le_{\ledger} y$. Thus, the set $C$ has to be $\LedgerDAG$-future-closed. Let $u\in \maximal{\ledger}{C}$. By definition of conflicting transactions, there exists some $x'$ and $u'$ such that $x\le_{\ledger} x'$, $u\le_{\ledger} u'$ and $x'$ and $u'$ are directly conflicting. Then by Definitions~\ref{def:conflicts} and~\ref{def:conflictingTx}, $u'\in \conflictset$. Note that $u'$ and $x$ are conflicting and, thus, $u'\in C$. Since $u\le_{\ledger} u'$ and $u'\in \maximal{\ledger}{C}$, we conclude that $u=u'\in \conflictset$.
\end{proof}

The set of relations between conflicts can be described with the notion of a Conflict Graph.

\begin{definition}[Conflict Graph]\label{def: Conflict Graph}
The \emph{Conflict Graph} $\ConflictGraph$ has vertex set $\conflictset$.  Two vertices in $\ConflictGraph$ are connected by an undirected edge if and only if the corresponding two conflicts are conflicting. 
\end{definition}

\begin{example}
We refer the reader to Figure~\ref{fig:conflicts} for an illustration of conflicts and the Conflict Graph. On the left part of the figure, we depict a UTXO DAG, where by coloring the box of a transaction, we indicate whether the transaction is a conflict. For instance, it can be seen that the orange transaction is directly conflicting with the red transaction, whereas the blue transaction is indirectly conflicting with the purple transaction. The relations between conflicts are demonstrated with the help of the corresponding Conflict Graph which is depicted on the right part of the figure.
\end{example}

\begin{figure}[t]
\centering
    \includegraphics[width=0.6\textwidth]{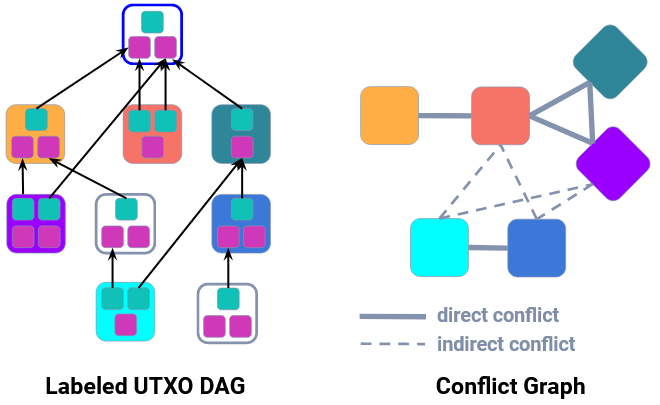}
    \caption{The UTXO DAG representation and the Conflict Graph} 
    \label{fig:conflicts}
\end{figure}

In contrast to the standard UTXO model, where no conflicts are allowed, they can be present in our generalization. We, however, require any two  conflicting transactions to be not comparable by the partial order $\le_{\ledger}$. The constraints in Assumption \ref{ass:constraints} about the addition of a transaction are relaxed:
\begin{ass}[Reality-based Ledger constraints]\label{ass:constraints2}
A transaction $x$ is added to the ledger $\ledger$ if it follows the following rules:
\begin{enumerate}[leftmargin=*]
    \item the transaction $x$ is syntactically correct;
    \item the sum of values of the inputs equals the sum of values of the outputs;
    \item the unlocking data $\unlock(x)$ is valid;
    \item $\inputs(x)$ are references to outputs which are not already consumed in $\cone{p}{\ledger}{x} \setminus \{ x\}$
\end{enumerate}
\end{ass}

\begin{remark}[Conflict-free past cone]\label{rem:conflictFree}
A  consequence of the $4$th point in Assumption~\ref{ass:constraints2} is that all past cones are conflict-free. In other words, we have that for every transaction $x\in\ledger$,  $\cone{p}{\ledger}{x}$  does not contain any pair of two conflicting transactions.
\end{remark}

\subsection{Labeled UTXO and Ledger DAGs and Conflict DAG}\label{sec:conflictDAGconflictGraph}

The existence of conflicts in the Ledger DAG plays a crucial role, as they eventually need to be resolved. In the following, we extract the necessary information for conflict resolution from the Ledger DAG. To this end, we labeled the transactions, see Definition~\ref{def:label}. We add this labeling to our UTXO and Ledger DAGs to keep  track of the conflicts and their dependencies.

\begin{figure*}[h!]
    \centering
    \includegraphics[width=.8 \textwidth]{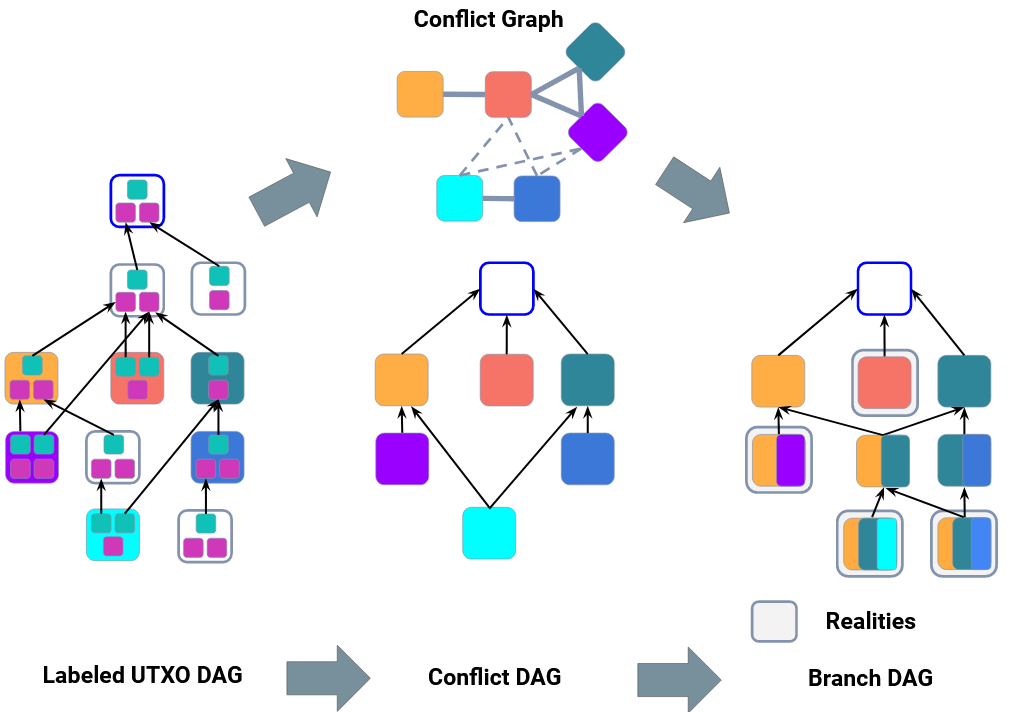}
    \caption{Derivation of the various graphs from the labeled UTXO DAG and the resulting realities (leaves of the Branch DAG). The colors represent the different labels of the conflicts. }
    \label{fig:DAGtransition}
\end{figure*}

\begin{definition}[Labeled UTXO and Ledger DAGs]\label{def:labeled UTXLedgerDAGs}
The \emph{labeled UTXO and Ledger DAGs} are the UTXO and Ledger DAGs with the  additional labels as described in Definition~\ref{def:label}.
\end{definition}

\begin{example} 
For an illustration of a labeled Ledger DAG, we refer the reader to Figures~\ref{fig:conflicts} and~\ref{fig:DAGtransition}. Note that the different colors correspond to the different labels in these figures. 
\end{example}
Now we define the restriction of the labeled Ledger DAG to the conflict set and the genesis using  Definition~\ref{def: minimal subDAG}.

\begin{definition}[Conflict DAG] \label{def: Conflict DAG}
The \emph{Conflict DAG}, written as $\ConflictDAG$, is defined as the minimal subDAG of $\LedgerDAG$ induced by the set of vertices $\conflictset\cup \{\orig\}$. 
\end{definition}

\begin{example}
For a more visual explanation of the above concept, we depict Figure~\ref{fig:DAGtransition}. Specifically, we demonstrate a UTXO DAG on the left part of the figure and the corresponding Conflict DAG in the middle. We note that the purple transaction in the Conflict DAG is not connected by an edge with the white one since there is a path connecting them which goes through the yellow transaction.
\end{example}
\begin{remark}
We observe that the Conflict DAG and the Conflict Graph represent only some partial information about the labeled UTXO DAG. Specifically, they are used to demonstrate different relations between the conflicts. Note that in general, it is not possible to construct the Conflict DAG using only the Conflict Graph and vice versa.
\end{remark}

Using the notion from Section~\ref{sec: graph-theoretical notions}, we denote the partial order on the set $\conflictset$ induced by $\ConflictDAG$ by $\le_{\conflictset}$. The past and future cones of a conflict $x\in\conflictset$ in the Conflict DAG are written as  $\cone{p}{\conflictset}{x}$ and $\cone{f}{\conflictset}{x}$, respectively. Finally, we give an observation saying that if a transaction is conflicting with some subset of conflicts $S$ (see Definition~\ref{def: conflicting sets}), it is possible to find a $\ConflictDAG$-minimal conflict in $S$ which is conflicting with that transaction.

\begin{proposition}[Transaction conflicting with a set]\label{prop: only minimal elements in branch}
Let a subset of conflicts $S\subseteq \conflictset$ be conflicting with a transaction $x\in \ledger$. Then there exists $c\in\min_{\conflictset}(S)$ such that $x$ and $c$ are conflicting.
\end{proposition} 

 \begin{proof}
By Definition~\ref{def: conflicting sets}, there exists some $y\in S$ such that $y$ and $x$ are directly or indirectly conflicting. From Definition~\ref{def: min and max elements}, it follows that there exists some $c\in\min_{\conflictset}(S)$  such that $c\le_{\conflictset} y$. Then by Definition~\ref{def:conflictingTx}, we conclude that $c$ is  conflicting with $x$.
 \end{proof}

\subsection{Branches and Branch DAG}\label{sec:branchDAG}

This section introduces the concepts of branches and Branch DAG, which help to handle the conflicting transactions.

\begin{definition}[Branch and set of branches]\label{def:branch}
A set of conflicts ${B}\subseteq \conflictset$ is called a \emph{branch} if and only if the two properties hold:
\begin{enumerate}[leftmargin=*]
    \item ${B}$ is conflict-free (cf. Definition~\ref{def: conflicting sets});
    \item ${B}$ is $\ConflictDAG$-past-closed (cf. Definition~\ref{def: future and past closed}).
\end{enumerate}
Define $\branchset$ to be the set of all branches, including  the so-called \emph{main branch} which represents the empty set. 
\end{definition}
Note that by Assumption~\ref{ass:constraints2}(4), for any conflict $x\in \conflictset$, the past cone $\cone{p}{\conflictset}{x}$ is a branch. In the following statement we discuss a sufficient condition for the union of branches to be a branch.

\begin{lemma}\label{lem:aggregateBranch}
Let ${B}_1,\dots,{B}_n\in \branchset$ be branches such that there exists a branch $A\in\branchset$ with ${B}_1,\dots,{B}_n\subseteq A$ . Then, the union  
\[{B}_1\cup\cdots\cup {B}_n\]
is also a branch, called the \emph{aggregate} branch of  ${B}_1,\dots,{B}_n$. 
\end{lemma}
\begin{proof}
Since, every ${B}_i$ is $\ConflictDAG$-past-closed, the union ${B}_1\cup\cdots\cup{B}_n$ is also $\ConflictDAG$-past-closed.  Since $ A$ contains no conflicting pairs, the union of  must not either.
\end{proof}

We proceed with a crucial observation saying that each branch can be represented as the aggregated branch of certain past cones.
\begin{lemma}[Aggregated branch]\label{lem: each branch is aggregate} 
Every branch ${B}$ can be uniquely written as the aggregated branch
$$
{B}= \cone{p}{\conflictset}{c_1} \cup \cdots \cup  \cone{p}{\conflictset}{c_n},
$$
where conflicts $c_1,\ldots,c_n$ are the $\ConflictDAG$-minimal elements in ${B}$. 
\end{lemma}

\begin{proof}
The branch  ${B}$ is finite. Hence, there exist unique  $\ConflictDAG$-minimal conflicts $c_1,\dots,c_n$ in ${B}$, i.e., for any $c_i$ and any other conflict $d\in {B}$, either it holds that $c_i\le_\C d$ or $d$ and $c_i$ are not comparable by the partial order.  Since ${B}$ contains the past cone of all its conflicts, it follows that ${B}$ is the aggregated branch of the branches $\cone{p}{\conflictset}{c_i}$. 
\end{proof}
\begin{remark}
Lemma~\ref{lem: each branch is aggregate} has some fundamental consequences for how we can implement the branches into the protocol. For instance, if a branch ${B}$ has a unique decomposition into 
 $$
{B}= \cone{p}{\conflictset}{c_1} \cup \cdots \cup  \cone{p}{\conflictset}{c_n},
 $$
 then we can set the branch ID of ${B}$ to be the hash of the  concatenation of the transaction IDs of the $c_i$'s (ordered in a canonical way).  Then, the branch ID of a branch of the form $\cone{p}{\conflictset}{c}$ with $c\in\conflictset$ is the same as the hash of the conflict $c$.  
 \end{remark}

A branch ${B}_1$ is called a \emph{subbranch} of branch ${B}_2$ if ${B}_1\subseteq {B}_2$. Lemma~\ref{lem: each branch is aggregate} already shows the recursive structure of the branches. 
This recursive structure can be encoded in the Branch DAG.

\begin{definition}[Branch DAG]\label{def:branchDAG}
Ordered by inclusion, $\branchset$ is a partially ordered set 
and defines a DAG. 
 Specifically, we put a directed edge from a branch $A$ to a branch $B$ if $A = {B}\cup c$ for some conflict $c\in \conflictset \setminus {B}$. The corresponding DAG is called the \emph{Branch DAG} and denoted by $\BranchDAG$.
\end{definition}

\begin{remark}\label{rem: complexity of Branch DAG}
We show an example of a Branch DAG in Figure~\ref{fig:DAGtransition}. In this example, the number of vertices in the Branch DAG is nine  which is larger than six, the number of conflicts. We note that in general, the number of vertices in a Branch DAG can be exponentially large in the number of conflicts. For instance, if there exist $t$ pairs of directly conflicting transactions such that any two transactions from different pairs are not conflicting, then the number of vertices in the Branch DAG is lower bounded by $2^t$. We refer the reader to Section~\ref{sec: numerical experiments}, where we explain how some functionalities based on the natural concept of a Branch DAG can be implemented in a more efficient way.
\end{remark}

Applying the notion from Section~\ref{sec: graph-theoretical notions}, we denote the partial order on the set $\branchset$ induced by $\BranchDAG$ by $\le_{\branchset}$; the set of parents and children of branch $B\in\branchset$ is written as $\parent{\branchset}{B}$ and $\child{\branchset}{B}$.

Conflicting transactions owe the existence of their conflict state to the presence of conflicts in their Ledger past cone. Since conflicts are labelled transactions, we can define a function that extracts all labels in this past cone. 
\begin{definition}[Maximal contained label set]\label{def:label set}
Let $\mathcal{Y}$ be the label space, and $\lab^{(p)}: \ledger \to 2^\mathcal{Y}$ be a function that for a given transaction $x\in \ledger$ returns all labels of the transactions in its Ledger past cone i.e., $\lab^{(p)}(x)=\{\lab(y): y\in \cone{p}{\ledger}{x}\}$. 
\end{definition}

\begin{remark} \label{rem: how to update the set of labels}
Practically this operation can be performed, e.g., through a graph search algorithm applied to the Ledger DAG (more computational intensive), or through a transaction-by-transaction record of conflict dependencies (more memory expensive). On one hand, for a given transaction $x$ we can identify all transactions with labels in the past cone $\cone{p}{\ledger}{x}$ by traversing the graph by means of  depth-first search. We can discontinue to search deeper than certain elements by cross-checking with the conflict set.  On the other hand, we can  inherit the maximal contained label set for a new arriving transaction from its parents and if a new conflict with $x$ is created, we traverse the future cone $\cone{f}{\ledger}{x}$ and update the maximal contained label set for all transactions there. 
\end{remark}

We can also define an equivalent function to obtain branch dependencies.
\begin{definition}[Maximal contained branch]
Let $\branchset$ be the set of all branches, and $\branch^{(p)}_{\ledger}: \ledger \to \branchset$ be a function that for a given transaction $x\in \ledger$ returns the maximal branch contained in $\cone{p}{\ledger}{x}$. 
\end{definition}
We note that there could not be two maximal branches in the Ledger past cone of a transaction (which is conflict-free) since, otherwise, we could consider their aggregate branch.
The above two definitions have the following correlation.
\begin{lemma}
The maximal contained label set of a transaction $x$ translates to the maximal branch that is contained in the past cone $\cone{p}{\ledger}{x}$. More precisely, we have that \[\lab^{(p)} (x)\setminus\{\bot\} =  \bigcup_{c \in \branch^{(p)}_{\ledger} (x)}  \lab(c).\]
\end{lemma}

\begin{proof}
By definition, $\lab^{(p)}(x)=\{\lab(y): y\in \cone{p}{\ledger}{x}\}$. The branch $\branch^{(p)}_{\ledger} (x)$ is included to $\cone{p}{\ledger}{x}$ and $\bot$ is not a label for any conflict by Definition~\ref{def:label}. This implies that 
$$
\bigcup_{c \in \branch^{(p)}_{\ledger} (x)}  \lab(c) \subseteq \lab^{(p)} (x) \setminus \{\bot\}. 
$$
Toward a contradiction, assume that the equality in the above formula does not hold, i.e., there exists some label $\ell\in \mathcal{Y}$ which is not present in the left-hand side. Consider the unique conflict $y\in \cone{p}{\ledger}{x}$ such that $\lab(y)=\ell$. We shall prove that this conflict should be included to the maximal contained branch. Indeed, the union $\cone{p}{\conflictset}{y}\cup \branch^{(p)}_{\ledger}(x)$ is $\ConflictDAG$-past-closed and is conflict-free as included to $\cone{p}{\ledger}{x}$. It follows that $y\in \branch^{(p)}_{\ledger}(x)$ and $\ell$ is present in the left-hand side of the displayed equation, which contradicts the assumption.
\end{proof}

\subsection{Realities in the  Branch DAG}\label{sec:realities}
In this section, we discuss maximal aggregated branches.
They are branches that present maximal acceptable valid versions of the ledger.

\begin{definition}[Maximal branch and reality]\label{def: reality}
A branch ${B}\in\branchset$ is \emph{maximal} if there exists no other branch $A\in\branchset$ such that  ${B}\subset A$. A maximal branch is called a \textit{reality}.
\end{definition}
Note that a reality always contains the main branch by definition since the empty set is included to all branches. An immediate consequence of the above definition is the following lemma.
\begin{lemma}\label{lem:realityLeaf}
The set of realities equals the set of leaves in the Branch DAG.
\end{lemma}
\begin{example}
Following the example depicted in Figure~\ref{fig:DAGtransition}, we observe that there are exactly four realities or leaves in the Branch DAG.
\end{example}

The following statement shows a link between realities and maximal independent sets in the Conflict Graph.
\begin{proposition}\label{prop:maximalIndReality}
There is a one-to-one correspondence between maximal independent sets of the Conflict Graph and realities.
\end{proposition}
\begin{proof}
 Let $I=\{c_1, \ldots, c_n \} \subseteq \conflictset$ be a maximal independent set in the Conflict Graph. We define the set  ${B}$ as follows
$$
     {B} := {B} (c_1, \ldots, c_n) :=  \cone{p}{\conflictset}{c_1} \cup \cdots \cup  \cone{p}{\conflictset}{c_n}.
$$
Since each past cone is $\ConflictDAG$-past-closed, the union is either. Assume that ${B}$ contains a conflicting pair of transactions, say $d$ and $e$. There must exist conflicts $c_i$ and $c_j$ such that $d \in \cone{p}{\conflictset}{c_i}$ and $e \in \cone{p}{\conflictset}{c_j}$. From Definition~\ref{def:conflictingTx} it follows that $c_i$  and $c_j$ are conflicting which implies a contradiction to the fact that $I$ is an independent set in the Conflict Graph.

The branch ${B}$ is also maximal. To see this assume the existence of a larger branch $A$ containing ${B}$, i.e., ${B}\subset A$ an let $d\in A\setminus{B}$ be a conflict from $A$ not included to ${B}$. Then $d$ and $c_1,\ldots,c_n$ are pairwise indirectly non-conflicting which contradicts the fact that the independent set $I$ is maximal.

Conversely, let ${B}=\{c_1,\ldots, c_n\}$ be a reality. Hence, $c_1,\ldots,c_n$ are not conflicting and $I:={B}$ is an independent set of the Conflict Graph. Toward a contradiction assume that $I$ is not maximal, i.e., there exists $d\in\conflictset \setminus {B}$ such that $d$ and $c_1,\ldots, c_n$ are pairwise non-conflicting. Define $A$ to be $  \cone{p}{\conflictset}{d} \cup \cone{p}{\conflictset}{c_1} \cup \cdots \cup  \cone{p}{\conflictset}{c_n}$. Clearly, $A$ is a branch containing ${B}$ as a subbranch. We arrive to a contradiction with Definition~\ref{def: reality}.
\end{proof}

\begin{definition}[Ledger of a reality]\label{def: R-ledger}
Let $R$ be a reality. We define the $R$-ledger as
\[
\ledger_r({R}) := \{ x\in \ledger: \ \lab^{(p)}(x) \subseteq \{\bot\} \cup \lab({R})\},
\]
where  $\lab({R}) := \bigcup\limits_{c\in {R}} \lab(c).$
\end{definition}

We can now give one of our main results as a direct consequence of the construction the Reality-based Ledger. 
\begin{theorem}\label{thm:Rledger}
Let $R$ be a reality. Then, $\ledger_r(R)$ is a consistent ledger (cf. Definition~\ref{def: consistent ledger}).
\end{theorem}
\begin{proof}
We have to prove that Assumption~\ref{ass:constraints2} implies Assumption \ref{ass:constraints} for every $R$-ledger. Constraints $1)-3)$ are trivially satisfied. Constraint $4)$ in Assumption~\ref{ass:constraints}  follows from the fact that a reality is a branch (and, hence, conflict-free) together with Constraint $4)$ in Assumption~\ref{ass:constraints2}.
\end{proof}

Similar to  a conflict-free ledger and Remark~\ref{rem: invariant ledger}, we provide invariance properties of the Reality-based Ledger.
\begin{remark}[Invariance properties of Reality-based Ledger]\label{rem: invariance reality}
The Reality-based Ledger upholds certain invariance properties and can thus be seen as an invariant data structure. More specifically, we notice that 
\begin{enumerate}[leftmargin=*] 
\item the Reality-based Ledger  might depend on the time parameter $t$, i.e. $\ledger=\ledger(t)$;
\item for any given reality $R\in\branchset$ at any given time $t$, it holds that the sum of output values of the state of the $R$-ledger remains constant
$$
\sum_{\substack{o=(\weight,\cond)\\o\in \state(\ledger_r(R))}} \weight = const.
$$
\end{enumerate}
\end{remark}

\section{Operations on the Reality-based Ledger}\label{sec: operations on the ledger}
In this section, we define several operations that we can perform on the Reality-based Ledger. In Section~\ref{sec:ConflictAddition}, we describe what happens when new conflicts are added to the ledger. In Section \ref{sec:RealitySelection} we provide an algorithm for the selection of a reality, i.e., a valid conflict-free ledger state, in the presence of a weight function imposed on transactions and branches. Finally, in Section \ref{sec:ConflictPruning} we describe how and when conflicts can be pruned to maintain a reasonable data  consumption for the operation of the Ledger. We refer to Figure \ref{fig:ProcessFlow} for an overview of the different operations and their dependencies.

\tikzset{%
            base/.style = {rectangle, rounded corners, draw=black,
                           minimum width=4cm, minimum height=1cm,
                           text centered, font=\sffamily},
          decision/.style = {base, diamond, inner sep=-5pt},
  activityStarts/.style = {base, fill=blue!10},
       startstop/.style = {base, fill=red!10},
    activityRuns/.style = {base, fill=green!30},
         processOpt/.style = {base, minimum width=2.5cm, fill=orange!10,
                           font=\ttfamily},
       process/.style = {base, minimum width=2.5cm, fill=orange!25,
                           font=\ttfamily,},
 processExt/.style = {base, minimum width=2.5cm, fill=green!10,
                           font=\ttfamily},                        
}

\begin{figure}
    \centering
\begin{tikzpicture}[node distance=2cm,
    every node/.style={fill=white, font=\sffamily}, align=center]
  \node (start)             [activityStarts]              {New incoming transaction};
  \node (isConflicting)     [decision, below of=start, yshift=-0.5cm]          {Is directly conflicting?};
 \node (updateConflictDAG) [process, below of = isConflicting, yshift=-0.5cm] {Update Conflict DAG; Algorithm \ref{alg: updating Conflict DAG}};
 \node (updateConflictGraph) [processOpt, below of = updateConflictDAG ] {Update Conflict Graph, Algorithm \ref{alg:updateConflictGraph} (optional)}; 
 \node (branchID) [processOpt, below  of = updateConflictGraph] { Update branches of transactions (optional)};
 \node (updateW) [processExt, below of = branchID] { Update weights, e.g. Examples \ref{ex:min} \ref{ex:minTS}, \ref{ex:extConsensus} (external)};
 \node (reality) [processOpt, below of = updateW] {Select Reality, Algorithm \ref{alg:selectionBranch} or \ref{alg:selectionConflictGraph} (optional)};
  \node (pruning) [process, below of = reality] {Prune of Conflict DAG and Conflict Graph, Algorithm \ref{alg: update UTXO DAG after confirmation}};
   
   \node (issue) [startstop, below of = pruning] { Issue  new transaction (optional)};
  \draw[->]             (start.south) -- (isConflicting);
  \draw[->]            (isConflicting.south) -- node {yes} (updateConflictDAG) ;
  \draw[->]           (updateConflictDAG) --  (updateConflictGraph) ;
  \draw[->]             (updateConflictGraph.south) -| (branchID.north) ;
  \draw[->]     (isConflicting.east)  --  node {no} ++(3,0) -- ++(0,-6) |-
                                    (branchID.east);
    \draw[->]             (branchID) -- (updateW);     
    \draw[->]             (updateW) --  (reality);
\draw[->]            (reality) -- (pruning)  ;
     \draw[->]            (pruning) --  (issue);               
\end{tikzpicture}
\caption{Flow of the different operations on the reality-based ledger. }
\label{fig:ProcessFlow}
\end{figure}
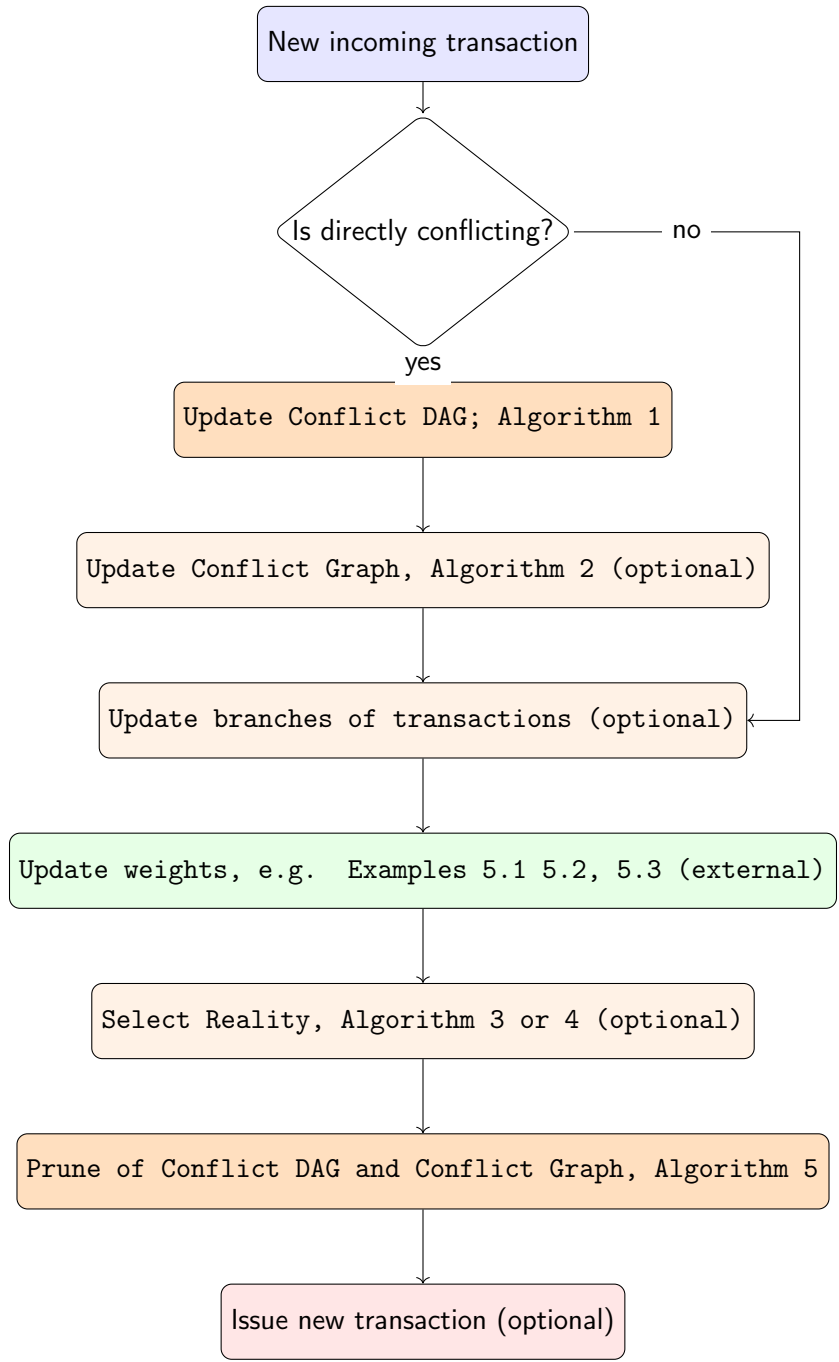

\subsection{Adding Conflicts}\label{sec:ConflictAddition}
In this section, we explain how new  conflicting transactions result in updating the Conflict DAG and the Conflict Graph. Throughout this section we assume  that there is only one new transaction $x$ which is a leaf in the Ledger DAG. In this case, the Ledger DAG is updated with only directed edges of the form $(x,y)$ for some $y\in \ledger$. 

First, we introduce a concept of closest conflicts in the past and future cones of a transaction which will be used for updating the Conflict DAG.

\begin{definition}[Closest conflicts in past and future cones]\label{def: closest conflicts}
For a transaction $z\in \ledger$, we define $\conflictsInCone{p}{\ledger}{z}$ to be $\min_{\conflictset}(S^{(p)})$ for $S^{(p)}:=\{y\in\conflictset\cup \{\orig\}:\ y<_{\ledger} z \}$, i.e., $S^{(p)}$ is the set of all conflicts in $\cone{p}{\ledger}{z}/\{z\}$.  Similarly, we define $\conflictsInCone{f}{\ledger}{z}$ to be $\max_{\conflictset}(S^{(f)})$ for $S^{(f)}:=\{y\in\conflictset:\ y<_{\ledger} z \}$, i.e., $S^{(f)}$ is the set of all conflicts in $\cone{f}{\ledger}{z}/\{z\}$.
\end{definition}

This notion resembles Definition~\ref{def:label set} and now we give a remark similar to Remark~\ref{rem: how to update the set of labels} on how to find these sets.  
\begin{remark}
Observe that $\conflictsInCone{p}{\ledger}{z}$  and   $\conflictsInCone{f}{\ledger}{z}$ can be obtained by first finding closest conflicts using breadth-first search (BFS) and reverse breadth-first search (RBFS). 
That means that in both cases we traverse $\LedgerDAG$ starting  at $z$ and stop traversing through transactions that are conflicts. Eventually, we identify the minimal/maximal elements in the obtained sets. 
\end{remark}

We provide a possible scheme to update the Conflict DAG in Algorithm~\ref{alg: updating Conflict DAG}.
Let $x$ be a new transaction and $Y\subseteq \ledger$ be the set of all transaction conflicting with $x$. Then, in Algorithm~\ref{alg: updating Conflict DAG}, we

\begin{enumerate}[leftmargin=*]
    \item update the set of vertices by adding the conflicts $x\cup Y$;
    \item add edges to the Conflict DAG using the notion of closest  conflicts in the past and future cones of transaction from $x\cup Y$;
    \item remove unnecessary edges in the Conflict DAG to keep it in the minimal form.
\end{enumerate} 

\begin{algorithm}[t]
\caption{Algorithm to update Conflict DAG}
\label{alg: updating Conflict DAG}
\KwData{Conflict DAG $\ConflictDAG = (\conflictset\cup \{\orig\}, E)$; new transaction  $x\in \ledger$ that is directly conflicting with transactions $Y\subseteq \ledger$}
\KwResult{updated Conflict DAG $\ConflictDAG = (\conflictset\cup \{\orig\}, E)$}
$\conflictset\gets \conflictset\cup\{x\}\cup Y$ \\

\For{$\forall y\in Y\cup \{x\}$}{
\For{$\forall v\in \conflictsInCone{p}{\ledger}{y}$}{
$E\gets E\cup \{(y,v)\}$\\
}
\For{$\forall v\in \conflictsInCone{f}{\ledger}{y}$}{
$E\gets E\cup \{(v,y)\}$
}
}

\For{$\forall y\in Y$}{
\For{$\forall p\in \parent{\conflictset}{y}$}{
\For{$\forall c\in \child{\conflictset}{p}$} {
\If{$c\in \cone{f}{\conflictset}{y}$}{
$E\gets E\setminus\{(c,p)\}$\label{it: main case}
}
}
}
}
\end{algorithm}
\begin{lemma}
The resulting graph $\ConflictDAG=(\conflictset\cup \orig, E)$ in Algorithm~\ref{alg: updating Conflict DAG} is the Conflict DAG as defined in Definition~\ref{def: Conflict DAG}.
\end{lemma}
\begin{proof}
Recall that the Conflict DAG is the minimal subDAG of the Ledger DAG induced by $\conflictset$. It is sufficient to check that we remove all unnecessary edges after adding correct edges using the notions $\conflictsInCone{p}{\ledger}{y}$ and $\conflictsInCone{f}{\ledger}{y}$.

Assume that some edge $(c,p)$, which was an edge in the original Conflict DAG, has to be removed to keep the Conflict DAG in the minimal form. This means that both $c$ and $p$ were already conflicts  such that $c<_{\conflictset} p$ with no other conflict between them  and now some conflict $y\in Y$ satisfies $c<_{\conflictset}y<_{\conflictset} p$.  Let $y^*$ be a $\ConflictDAG$-minimal conflict among all  $y$  satisfying the latter inequality, i.e, $c\in \child{\conflictset}{y^*}$. One can check that all such edges $(c,p)$ are removed from the Conflict DAG in line~\ref{it: main case} of Algorithm~\ref{alg: updating Conflict DAG}.
\end{proof}

Assume that Algorithm~\ref{alg: updating Conflict DAG} is already completed. In Algorithm~\ref{alg:updateConflictGraph} we describe a possible procedure to update the Conflict Graph.  In this algorithm, we  only add new edges to the Conflict Graph $\ConflictGraph$. Specifically, all conflicts in the future cones of conflicts in $Y$ become adjacent with $x$ in $\ConflictGraph$. In addition, all conflicts in $x\cup Y$ inherit $\ConflictGraph$-neighbours from their parents in $\ConflictDAG$.

\begin{algorithm}[t]
\caption{Algorithm to update Conflict Graph}
\label{alg:updateConflictGraph}
\KwData{Conflict Graph $\ConflictGraph = (\conflictset_{old}, E)$; new transaction  $x\in \ledger$ that is directly conflicting with transactions $Y\subseteq \ledger$; updated Conflict DAG $\ConflictDAG = (\conflictset\cup \orig, E)$;}
\KwResult{updated Conflict Graph $\ConflictGraph = (\conflictset, E)$}
\For{$\forall y\in Y$}{
\For{$\forall z\in \cone{f}{\conflictset}{y}$}{
$E\gets E\cup \{(x,z)\}$\label{algUpdateConflictGraph: x is directly conflicting}
}
}

\For{$\forall y\in Y\cup\{x\}$}{
{\scriptsize{\Comment*[l]{assume $\ConflictDAG$-descending order}}}
\For{$\forall p\in \parent{\conflictset}{y}$}{
\For{$\forall z\in N_{\conflictset}(p)$}{
{\scriptsize{\Comment*[l]{$N_{\conflictset}(p)$ denotes the set of neighbours of $p$ in $\ConflictGraph$}}}
$E\gets E\cup \{(y,z)\}$ \label{algUpdateConflictGraph: remaining}
}
}
}
\end{algorithm}
\begin{lemma}
The resulting graph $\ConflictGraph=(\conflictset, E)$ in Algorithm~\ref{alg:updateConflictGraph} is the Conflict Graph as defined in Definition~\ref{def: Conflict Graph}.
\end{lemma}
\begin{proof}
By Definition~\ref{def: Conflict Graph}, if two conflicts $c_1,c_2\in \conflictset$ are connected by an edge in the Conflict Graph, then there exist $e_1,e_2\in \conflictset$ with $c_1\le_{\conflictset} e_1$ and $c_2\le_{\conflictset} e_2$ such that $e_1$ and $e_2$ are directly conflicting. Let $E'$ denote the set of edges to be included to the Conflict Graph after transaction $x$ arrives. We claim that any edge from $E'$ should contain at least one conflict from the set $x\cup Y$. This is true since $x$ is a leaf in both the Ledger DAG and the Conflict DAG.

Let us start with considering edges from $E'$ of type $(x,z)$ for some $z\in\conflictset$ such that there exists $y\in \conflictset$ with $z\le_{\conflictset} y$, and such that $x$ and $y$ are directly conflicting. 
Clearly, $y$ has to be from $Y$ since $x$ is directly conflicting with transactions from $Y$ only.
In line~\ref{algUpdateConflictGraph: x is directly conflicting} of Algorithm~\ref{alg:updateConflictGraph}, by traversing over all transactions $z$ in the Conflict DAG that are contained in the future cones of conflicts from $Y$, we add all such edges $(x,z)$ to $E$.

Observe that there is no edges in $E'$ of type $(y,z)$ for $y\in Y$ and $z\in \conflictset\setminus\{x\}$ such that
there exists $w\in \conflictset$ with $z\le_{\conflictset} w$  and such that $y$ and $w$ are directly conflicting. That is true since we do not include any edge in the Ledger DAG that contains $y\in Y$ when transaction $x$ arrives.

Thus, it remains to add edges of type $(y,z)$ with $y\in Y\cup \{x\}$ such that there exist $e_1,e_2\in \conflictset$ with $y<_{\conflictset} e_1$ and $z<_{\conflictset} e_2$ such that $e_1$ and $e_2$ are directly conflicting. Note that the above inequalities in the partial order relations are strict and we can utilize the notion of parents in $\ConflictDAG$ to add edges recursively. Thereby, by assuming some  $\ConflictDAG$-descending order over the set $Y\cup\{x\}$, we iteratively perform the following step in~line~\ref{algUpdateConflictGraph: remaining} of~Algorithm~\ref{alg:updateConflictGraph}.  For every conflict  $y\in  Y\cup\{x\}$, the set of $\ConflictGraph$-neighbors of $y$  is updated by looking at $\ConflictGraph$-neighbours of parents  of $y$ in $\ConflictDAG$ as follows 
$$
N_{\conflictset}(y)\gets N_{\conflictset}(y) \cup \left\{\bigcup_{p\in\parent{\conflictset}{y}} N_{\conflictset}(p)\right\}.
$$
\end{proof}

\subsection{Reality Selection Algorithms}\label{sec:RealitySelection}

In a system with a conflict-free UTXO Ledger, the account owners can learn about their current balance by inspecting the state of an instance of a ledger, see Section \ref{sec: Standard UTXO}. In the case of the Reality-based Ledger this is more complicated, since only a single reality represents a valid conflict-free ledger, and thus can yield a valid ledger state.
It is, therefore, up to the operator of the ledger instance, to choose which ledger state to evaluate to inform the account owners about their balance. Alternatively, and in a more trustless fashion, the account owner and the operator of the ledger instance constitute the same entity. 

In this section, we propose reality selection algorithm that construct a preferred reality by utilizing a weight function for transactions. We impose several natural constraints on the weight function.
\begin{ass}\label{ass: weight function on transaction}
We assume that there exists a weight function on the set of transactions $\AW: \ledger\to [0,1]$ which satisfies the following properties
\begin{enumerate}[leftmargin=*]
\item unitarity: $\AW (\orig)=1;$
\item monotonicity: for any two transactions  $x, y\in \ledger$ such that $x \le_{\ledger} y$, it holds that 
    $$
    \AW (x) \le \AW(y);
    $$
\item consistency: let $x_1,\ldots, x_s$ be pairwise conflicting transactions.\footnote{We say that transactions  $S\subseteq \ledger$ are pairwise conflicting if any pair of transactions $x, y\in S$ are conflicting. } Then it holds that 
$$
\sum_{i=1}^{s}\AW(x_i) \le 1.
$$
\end{enumerate} 
\end{ass}
We naturally extend the domain of the function $\AW$ to the set of all branches $\branchset$ as follows. For a branch $B\in\branchset$, we define $\AW(B)$ to be 
$$
\AW\left(B\right):=\min_{x\in B}\AW(x).
$$

\begin{remark}\label{rem: weight function on branches}
A weight function $\AW$ induces a weight function on the set of branches $\branchset$ with the following monotonicity property: let ${B}_1, {B}_2\in \branchset$ such that ${B}_1 \subseteq {B}_2$, it holds that 
\[
    \AW ({B}_1) \geq \AW({B}_2).
  \]  If the weight function satisfies Assumption \ref{ass: weight function on transaction}, we also have that  $\AW ({B}_1) + \AW({B}_2)\leq 1$ for any two conflicting branches ${B}_1, {B}_2\in \branchset$. We also observe that for a conflict $c\in\conflictset$, $\AW(c)=\AW(\cone{p}{\conflictset}{c})$.
\end{remark}

Let us give some examples of weight functions satisfying Assumption \ref{ass: weight function on transaction}.

\begin{example}[Minimal hash]\label{ex:min}
Transactions can be ordered, relatively to each other, through several means. One such way is to hash the content of the transaction, see Remark \ref{rem: hash function}. Here we assume the existence of such a function $\hash: \ledger \to \{0,1\}^h$. We define the weight function using the following steps:
\begin{enumerate}[leftmargin=*]
    \item $\AW(\orig)=1$;
    \item for every $x \in \C$ we set $\AW(x)=1$ if 
    \[
    \hash(x)= \min_{ y \in N_{\conflictset}(x)} \hash(y),
    \]
    and $\AW(x)=0$ otherwise;
    \item inductively starting from the genesis for every $x \in \ledger \setminus \C$
     we set 
    \[\AW(x)= \min_{p \in \parent{\ledger}{x}} \AW(p)\] and update  the weight of $x\in  \C$ to 
    \[\AW(x)= \min_{p \in \parent{\ledger}{x}} \AW(p)\] if $\AW(x)$ was set to $1$ in the second step. 
\end{enumerate}
\end{example}
\begin{example}[Minimal timestamp]\label{ex:minTS}
Transactions can carry additional information, for instance, timestamps. These can be used to decide between two conflicting transactions. Replacing the hashes by timestamps in Example~\ref{ex:min} yields a weight function based on the timestamps of the transactions.
\end{example}
The two examples above make the most sense in a distributed system under the assumptions of eventual consistency. For applications in the DLT space, more efficient and robust weights are appropriate.

\begin{example}[External consensus]\label{ex:extConsensus}
An external consensus protocol can determine the weights. For example, nodes could agree on the weights via additional direct communication and employ Byzantine-fault-tolerance mechanisms. Weights can also be inherited by the data structure that carries the ledger; for instance, $\AW(x)=1$ if $x$ is contained in the longest chain, \cite{nakamoto2008bitcoin}, or in the heaviest subtree, \cite{GHOST}; and $\AW(x)=0$ otherwise. Finer weights can be obtained using the distance between the longest and second-longest chain. 
\end{example}

\begin{example}[Approval weight]
The monotonicity property, see  Assumption~\ref{ass: weight function on transaction}, suggests that we can define the weights  recursively using the underlying DAG structure of the  DLT. This definition enables an internal consensus mechanism on the weights and, therefore, on the preferred reality. These ideas are expanded in detail in~\cite{OTV}.
\end{example}

To determine which reality an operator of the ledger should prefer, we propose to perform the recursive exploration algorithm described in Algorithm~\ref{alg:selectionBranch}. In this algorithm, we start at the main branch of the Branch DAG and walk on this graph until we reach a leaf. The algorithm prefers to go to the child with the highest value of the weight function.
We observe that the resulting ${R}$ is a maximal branch or a reality by construction. One can readily see that the provided algorithm has reasonable complexity despite the fact that the Branch DAG can be exponentially large in the number of conflicts (cf. Remark~\ref{rem: complexity of Branch DAG}). Indeed,
the number of iterations in the while-loop is bounded by the depth of the Branch DAG which is at most $|\conflictset|$. 
The number of elements in $\child{\branchset}{R}$ is also bounded by $|\conflictset|$.
Thereby, the complexity of Algorithm~\ref{alg:selectionBranch} can be estimated as $O(|\conflictset|^2)$. These observations are summarized below.

\begin{algorithm}[t]
\caption{Reality selection in Branch DAG}
\label{alg:selectionBranch}
\KwData{Branch DAG $\BranchDAG = (\branchset, E)$}
\KwResult{reality ${R}\in \branchset$}
${R}\gets \emptyset$ {\scriptsize{\Comment*[r]{{main branch in $\BranchDAG$}}}}
\While{${R}$ is not a leaf in $\BranchDAG$}{
${B}^* \gets \argmax\{\AW({B}): {B} \in \child{\branchset}{R} \}$ {\scriptsize{\Comment*[r]{use $\min\hash({B}\setminus {R})$ for breaking ties }}}
${R}\gets {B}^*$
}
\end{algorithm}
\begin{proposition}\label{prop: complexity of algorithm}
The resulting set ${R}$  in Algorithm~\ref{alg:selectionBranch} is a reality. The complexity of this algorithm is $O(|\conflictset|^2)$.
\end{proposition}

As there is a one-to-one correspondence between realities and maximal independent sets of the Conflict Graph by Proposition~\ref{prop:maximalIndReality}, we propose an alternative reality selection procedure based on the Conflict Graph in Algorithm~\ref{alg:selectionConflictGraph}. In this algorithm, we start with the empty set and iteratively construct a subset $R$ of conflicts. Specifically, we add a conflict to this set if this conflict is not conflicting with $R$ and attains the highest value of the weight function. By construction, Algorithm~\ref{alg:selectionConflictGraph}  leads to a maximal independent set in the Conflict Graph or a reality in the Branch DAG. The number of iterations in the while-loop is bounded by $|\conflictset|$ and the number of $\ConflictGraph$-neighbours is also bounded by $|\conflictset|$. Thus, it is possible to implement this algorithm with complexity $O(|\conflictset|^2)$.  In the following statement we verify that the outcomes of the two algorithms coincide.

\begin{algorithm}[t]
\caption{Reality selection in Conflict Graph}
\label{alg:selectionConflictGraph}
\KwData{Conflict Graph $\ConflictGraph=(\conflictset,E)$}
\KwResult{reality ${R}\in \branchset$}
${R}\gets \emptyset$\\
$U\gets \conflictset$\\
\While{$|U|\neq 0$}{ \label{line: minimal hash among maximal remaining elements}
$c^*\gets \argmax \{ \AW(c): c\in \max_{\conflictset}(U)\}$ {\scriptsize{\Comment*[r]{use $\min\hash(c)$ for breaking ties}}}
${R} \gets {R} \cup \{c^*\}$\\ $U \gets U  \setminus \{N_{\conflictset}(c^{*}) \cup\{c^{*}\}\}$
}
\end{algorithm}

\begin{theorem}
Algorithms~\ref{alg:selectionBranch} and~\ref{alg:selectionConflictGraph} provide the same reality as output.
\end{theorem}

\begin{proof}
The proof is done by induction on the number of iterations in the while-loops. Both algorithms start with the empty set and add at the first step the same conflict $c^{*}$, namely the one achieving the highest value of the weight function $\AW(\cdot)$. Note that $c^{*}$ is $\ConflictDAG$-maximal in the set of conflicts $\conflictset$ and represents a child of the main branch in $\BranchDAG$.

Let us assume that both algorithms constructed the same branch ${R}$ after some number of steps. 
Then on one hand, Algorithm~\ref{alg:selectionConflictGraph} will pick the conflict $c^*$ with the highest value of $\AW(\cdot)$ among all $\ConflictDAG$-maximal elements in the set $U$, where $U$ is the set of conflicts in $\conflictset\setminus {R}$ that are not conflicting with ${R}$. 
On the other hand, Algorithm~\ref{alg:selectionBranch} will pick a branch ${B}^*$ with the highest value of $\AW(\cdot)$ over all children of ${R}$ in~$\BranchDAG$. It remains to show that ${B}^*$ obtained in Algorithm~\ref{alg:selectionBranch} coincides with  ${R}\cup \{c^*\}$ obtained in Algorithm~\ref{alg:selectionConflictGraph}. 

First, we prove that ${R}\cup \{c^*\}$ is a branch (cf. Definition~\ref{def:branch}).  The set ${R}\cup \{c^*\}$ does not contain conflicting transactions since all transactions conflicting with ${R}$ were removed from $U$ at the previous steps  and $c^{*}\in U$. Seeking a contradiction assume that ${R}\cup \{c^*\}$ is not $\ConflictDAG$-past-closed. Since ${R}$ is a branch by the inductive hypothesis, it may happen only when there exists some $b\in \cone{p}{\conflictset}{c^{*}}$ such that $b\not\in {R}$. From $c^*\in U$ it follows that $b\in U$. Since $c^*\in \max_{\conflictset}(U)$ and $c^*\le_{\conflictset} b$, we conclude that $c^*=b$. Thus,  ${R}\cup \{c^*\}$ is indeed a branch which belongs to the set of children $\child{\branchset}{R}$ by Definition~\ref{def:branchDAG}.

Since ${B}^*\in \child{\branchset}{R}$, the branch ${B}^*$ can be represented as $d\cup {R}$ for some $d\in\conflictset$. The set ${B}^*$ is a branch and does not contain conflicting transactions and, thus, we have that $d\in U$. Moreover, $d\in \max_{\conflictset}(U)$ since ${B}^*$ is $\ConflictDAG$-past-closed. 

Now we will show that 
\begin{equation}\label{eq: weight of d is equal}
\AW(d)=\AW({B}^*).
\end{equation}
By definition of the weight function $\AW({B}^*)=\min\{\AW(c) : c\in d\cup {R}\}$. Seeking a contradiction, assume that the minimum is attained at some $e \neq d$, i.e., $\AW(e)<\AW(d)$. Let $U'$ and ${R}'$ be the sets $U$ and ${R}$ right before the element $e$ was included to ${R}$. Then one can find $d'\in\cone{p}{\conflictset}{d}$ such that $d'\in\max_{\conflictset}(U')$. By the monotonicty of the weight function, we get that $\AW(d')\ge \AW(d)>\AW(e)$ which contradicts the fact that $e$ was included to $U'$ and, thus, achieves the maximum of the weight function over all conflicts in $\max_{\conflictset}(U')$.

Using similar ideas as above one can prove that 
\begin{equation}\label{eq: weight of c^* is smaller}
\AW(c^*)\le \AW({R}).
\end{equation}

Finally, it remains to show that $\AW(d)=\AW(c^*)$. By definition of $c^{*}$ in Algorithm~\ref{alg:selectionConflictGraph}, $\AW(c^*)\ge \AW(d)$ as both $c^*$ and $d$ belong to $\max_{\conflictset}(U)$. On the other hand, by definition of ${B}^{*}$ in Algorithm~\ref{alg:selectionBranch}, $\AW({B}^{*})\ge \AW(c^*\cup {R})$. Combining the latter inequality with~\eqref{eq: weight of d is equal}-\eqref{eq: weight of c^* is smaller} implies
\[
\AW(d)=\AW({B}^*)\ge \AW(c^*\cup {R})=\AW(c^*),
\]
which leads to $\AW(d)=\AW(c^*)$. To ensure that $d=c^*$ we observe that both algorithms use $\min \hash(\cdot)$ for breaking ties.
\end{proof}

\begin{example}
In Figure~\ref{fig:RealitySelection}, we depict an illustrative example that demonstrates how both reality selection algorithms work. We make use the same labeled UTXO DAG, the Conflict DAG and the Conflict Graph as in Figure~\ref{fig:DAGtransition}. First, we note that there are three iterations in the while-loops of the both algorithms. Conflicts in Figure~\ref{fig:RealitySelection} are represented by colorful boxes. The value of the function $\AW(\cdot)$ is depicted inside the boxes. The selected set of conflicts ${R}$ at every step is highlighted by green borders. At the first step both algorithms include to ${R}$, which was initialized as the empty set, the yellow conflict which has the highest weight $0.7$. Since the red conflict is conflicting with the yellow conflict, it should be removed from the set $U$. At the second step, Algorithm~\ref{alg:selectionBranch} takes the child of ${R}$ in $\BranchDAG$ that has the highest weight. This child is a branch that consists of two conflicts, the yellow one and the aquamarine one. At the same moment Algorithm~\ref{alg:selectionConflictGraph} finds the conflict with the highest weight in the set of $\ConflictDAG$-maximal elements of $U$. This set consists of the purple conflict and the aquamarine conflict and the latter has the highest weight. Finally, at last step, both algorithms update set ${R}$ by adding the blue conflict.
\end{example}

\begin{figure*}[t]
    \centering
    \includegraphics[width=0.8\textwidth]{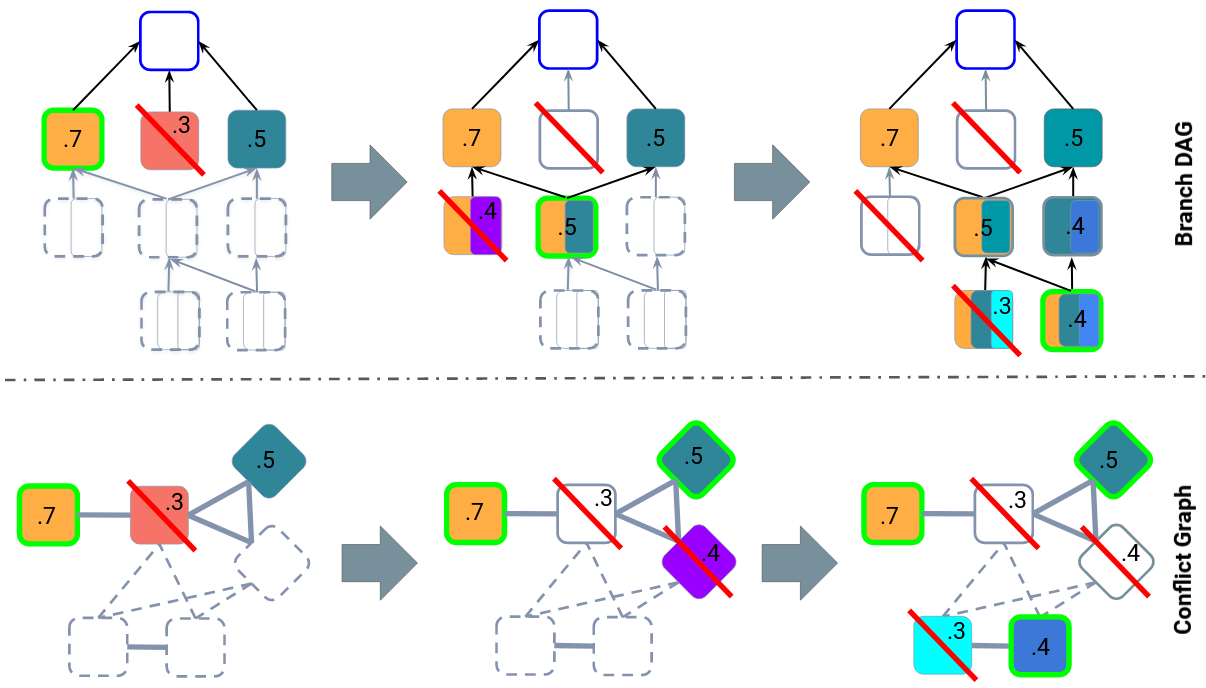}
    \caption{Demonstrative example of the reality selection algorithms based on the Branch DAG and the Conflict Graph.}
    \label{fig:RealitySelection}
\end{figure*}
\subsection{Pruning Conflicts}\label{sec:ConflictPruning}
In this section, we explain how we prune conflicting transactions and update the Conflict DAG, the Conflict Graph and the Ledger DAG when the weight function $\AW$ of a certain branch exceeds a given threshold.

\begin{definition}[Confirmed branch]
Let $\theta \in (0.5, 1]$ be a fixed threshold. A branch ${B}\in\branchset$ is called $\AW$-\textit{confirmed} if 
$\AW({B})\ge\theta$. 
\end{definition}

\begin{lemma}\label{lem: the only minimal confirmed branch}
Let $\branchset'\subseteq\branchset$ be the set of $\AW$-confirmed branches. Then there exists a unique $\BranchDAG$-minimal branch in the set $\branchset'$. In other words, the subDAG of $\BranchDAG$ induced by $\branchset'$ has  a unique leaf.
\end{lemma}

\begin{proof}
Seeking a contradiction, assume the existence of two $\BranchDAG$-minimal confirmed branches ${B}_1$ and ${B}_2$. If these branches are conflicting, then $\AW({B}_1)+\AW({B}_2)\le 1$ by Remark~\ref{rem: weight function on branches} and we come to a contradiction as $\AW({B}_1)+\AW({B}_2)\ge2\theta>1$.  Thus, the union of these branches does not contain a pair of conflicting transactions and ${B}:={B}_1\cup {B}_2$ is also a branch. Moreover,  ${B}$ is $\AW$-confirmed as
$$
\AW({B})=\min_{x\in{B}_1\cup{B}_2} \AW(x)\ge \min(\AW({B}_1),\AW({B}_1))\ge \theta.
$$
Clearly, ${B}<_{\branchset}{B}_1$ and  ${B}<_{{\branchset}}{B}_2$ and we arrive to a contradiction with the fact that ${B}_1$ and ${B}_2$ are $\BranchDAG$-minimal in the set of confirmed branches $\branchset'$.
\end{proof}
Let  ${B}$ be the unique $\BranchDAG$-minimal  $\AW$-confirmed branch as in Lemma~\ref{lem: the only minimal confirmed branch}. We propose to update all graph structures such that the conflicts in $B$ are no longer conflicts, i.e., all transactions conflicting with $B$ should be removed from the ledger $\ledger$. A possible way to achieve that is described in Algorithm~\ref{alg: update UTXO DAG after confirmation}.

In this algorithm, we first identify  the set of conflicts $\conflictset'$ conflicting with $B$ and the set of conflicts $\conflictset''$ that should be removed from $\conflictset$. Note that $B\cup \conflictset' \subseteq \conflictset''$. 
Then we
\begin{enumerate}[leftmargin=*]
    \item remove all edges adjacent to $\conflictset'$ from the Conflict Graph $\ConflictGraph$. We note that $\conflictset''$ is the set of isolated vertices in $\ConflictGraph$. We remove $\conflictset''$  from $\ConflictGraph$.
    \item prune the Ledger DAG by removing all transactions that are directly or indirectly conflicting with ${B}$ and all edges adjacent to them by traversing future cones of $\conflictset'$;
    \item remove all vertices corresponding to $\conflictset''$  and all edges adjacent to $\conflictset''$ from the Conflict DAG. If necessary, we add some edges to the Conflict DAG to make it connected.
\end{enumerate} 

Recall that the Conflict DAG and the Conflict Graph are concepts which are derived from the Ledger DAG, i.e., $\ConflictDAG = \ConflictDAG(\LedgerDAG)$ and $\ConflictGraph = \ConflictGraph(\LedgerDAG)$. Thereby, one needs to show consistency between the resulting three graphs.

\begin{algorithm}[ht]
\caption{Algorithm to update Ledger DAG, Conflict DAG and Conflict Graph  after branch confirmation}
\label{alg: update UTXO DAG after confirmation}
\KwData{Ledger DAG $\LedgerDAG=(\ledger,E_{\LedgerDAG})$, Conflict DAG $\ConflictDAG=(\conflictset\cup\{\orig\},E_{\ConflictDAG})$, Conflict Graph $\ConflictGraph = (\conflictset, E_{\ConflictGraph})$, confirmed branch ${B}$}
\KwResult{updated Ledger DAG $\LedgerDAG=(\ledger,E_{\LedgerDAG})$, Conflict DAG $\ConflictDAG=(\conflictset\cup\{\orig\},E_{\ConflictDAG})$, Conflict Graph $\ConflictGraph = (\conflictset, E_{\ConflictGraph})$}
$M\gets \min_{\conflictset}({B})$\\
$\conflictset'\gets \emptyset$ {\scriptsize{\Comment*[r]{conflicts conflicting with $B$}}}
$\conflictset''\gets \emptyset$ {\scriptsize{\Comment*[r]{conflicts to be removed from $\conflictset$}}}

\For{$\forall x\in M$}{
\For{$\forall y\in N_{\conflictset}(x)$}{
$\conflictset'\gets \conflictset' \cup \{y\}$\\
\For{$\forall z\in N_{\conflictset}(y)$}{
$E_{\ConflictGraph}\gets E_{\ConflictGraph}\setminus\{(y,z)\}$}}}

\For{$\forall x\in \conflictset$}{
\If{$x$ is isolated in $\ConflictGraph$}{
$\conflictset''\gets \conflictset'' \cup\{x\}$
}
}
\For{$\forall y\in \max_{\conflictset}(\conflictset ')$}{
\For{$\forall z\in \cone{f}{\ledger}{y}$}{
{\scriptsize{\Comment*[l]{assume $\LedgerDAG$-descending order}}}
$\ledger\gets \ledger\setminus\{z\}$\\
\For{$\forall p\in \parent{\ledger}{z}$}{
$E_{\LedgerDAG}\gets E_{\LedgerDAG} \setminus\{(z,p)\}$
}
}
}
\For{$\forall y\in \conflictset ''$}{
\For{$\forall p\in \parent{\conflictset}{y}$}{
$E_{\ConflictDAG}\gets E_{\ConflictDAG} \setminus\{(y,p)\}$
}
\If{$y\in \max_{\conflictset}(\conflictset'')$}{
\For{$\forall c\in\child{\conflictset}{y}$}{
$E_{\ConflictDAG}\gets E_{\ConflictDAG} \setminus\{(c,y)\}$\\
\If{$\parent{\conflictset}{c}\subseteq \conflictset''$}{
$E_{\ConflictDAG}\gets E_{\ConflictDAG} \cup\{(c,\orig)\}$
}
}
}
$\conflictset \gets \conflictset \setminus \{y\}$\\
}
\end{algorithm}
 \begin{lemma}
 The resulting Conflict DAG and the resulting Conflict Graph  are consistent with the resulting Ledger DAG in Algorithm~\ref{alg: update UTXO DAG after confirmation}.
 \end{lemma}

\begin{proof}
By Proposition~\ref{prop: only minimal elements in branch}, to identify the set of all conflicts that are conflicting with ${B}$, it suffices to consider the $\ConflictDAG$-minimal elements in ${B}$. In Algorithm~\ref{alg: update UTXO DAG after confirmation}, we set $M$ to be the set of these minimal conflicts. Then we construct the set $\conflictset'\subseteq \conflictset$ of conflicts that are conflicting with branch ${B}$ by looking at the neighbours of $M$ in the Conflict Graph. Similar to Proposition~\ref{prop: conflicting transactions}, one can check that the set $\conflictset'$ is $\ConflictDAG$-future-closed.  Let $\ledger'\subseteq \ledger$ denote the set of transactions that are conflicting with ${B}$. By Proposition~\ref{prop: conflicting transactions}, $\ledger'$ is $\LedgerDAG$-future-closed and $\max_{\ledger}(\ledger')=\max_{\conflictset}(\conflictset')$. In Algorithm~\ref{alg: update UTXO DAG after confirmation}, we 
\begin{enumerate}[leftmargin=*]
     \item remove all edges adjacent to $\conflictset '$ from the Conflict Graph and construct the set  $\conflictset''$ of isolated vertices in the updated Conflict Graph. Clearly, all conflicts in $\conflictset''$ are no longer conflicts;
    \item update the Ledger DAG by recursive traversing the Ledger future cone  of $\max_{\conflictset}(\conflictset')$ and removing transactions $\ledger'$ and all edges adjacent to them;
    \item update the Conflict DAG by removing all conflicts  $\conflictset''$ from $\conflictset$ and all edges adjacent to them. It is possible that after performing this step some vertices in $\conflictset\setminus \conflictset''$ have out-degree zero. To make the Conflict DAG connected again, we add edges from all such conflicts to the genesis $\orig$.
\end{enumerate}

\end{proof}
We conclude this section with a sufficient condition for the pruned data structure to be again conflict-free.
\begin{theorem}
Let $\ledger$ be a ledger and ${B}$ be a  $\AW$-confirmed reality. Then, the set of transactions in the pruned Ledger DAG resulting from Algorithm~\ref{alg: update UTXO DAG after confirmation} is conflict-free and coincides with the $B$-ledger $\ledger_r(B)$ as in Definition~\ref{def: R-ledger}.
\end{theorem}

\begin{proof}
Suppose ${B}$ is a reality. Then, we note that the set of conflicts to be removed from $\conflictset$  is $\conflictset'' = \conflictset$. Indeed, since $B$ is a maximal independent set in the Conflict Graph (cf. Proposition~\ref{prop:maximalIndReality}), all conflicts in $\conflictset \setminus B$ are conflicting with $B$  and has to be removed from the Ledger DAG. Thus, the conflicts in $B$ are no longer conflicts and become ordinary transactions. Thus, the resulting set of transactions does not contain a pair of directly conflicting transactions. In addition, all transactions $x$ such that $\lab^{(p)}(x)\subseteq \lab^{(p)}(B)$ are kept in the set of transactions.
\end{proof}

\section{Numerical Experiments}\label{sec: numerical experiments}
We have implemented a standalone program\footnote{available at \url{https://github.com/nikitapolyanskii/reality-ledger}} in Go 1.20 that provides several functionalities described in the paper. We have conducted all benchmarks on Windows 11 with a single CPU of Intel Core i7-11370H with 3.30GHz, and 8GB of memory. The simulation results presented in this section provide a lower bound on the number of transactions that can be processed by a node with hardware of this kind, provided that one CPU is dedicated to managing a reality-based ledger and the reality-based ledger is completely stored in RAM. While it is natural to parallelize transaction processing for a reality-based ledger, the implementation aspect of this question is beyond the scope of the paper. We would like to note that our benchmarks only address updating all basic data structures and do not include signature checks or transaction executions.
In our numerical experiment, we have generated a stream of pseudo-random UTXO transactions that meet specific criteria. Specifically, we have used the following guidelines:
\begin{itemize}
   \item The number of inputs for a new transaction is randomly sampled from a uniform distribution on the set $\{1,2\}$.
   \item For non-conflicting transactions, inputs are selected randomly and uniformly from the set of all unspent outputs. 
   \item A new transaction becomes conflicting (or directly conflicting with an existing transaction) with a given probability of $p_{\textrm{conflict}}\in\{0.01,0.05,0.1, 0.5\}$. To accomplish this, we randomly select one input label from the set of already consumed outputs and the remaining input (if there is one) from the unspent outputs. 
    \item The number of outputs for a new transaction is randomly sampled from a uniform distribution on the set $\{1,2,3\}$, with each new output created as a hash digest (using SHA256) of a random value.    
    \item The genesis has $16$ outputs and $0$ inputs.
\end{itemize}
It is expected that the number of conflicts using this model will fall between $N p_{\textrm{conflict}}$ and $2N p_{\textrm{conflict}}$ with high probability, where $N$ is the total number of transactions. This is because a new transaction can cause a previously non-conflicting transaction to directly conflict with the new one.
\begin{remark} Let us comment on the choice of the above values. We investigated all transactions, around $300k$, of the Shimmer main net from 2022/09/27 to 2023/01/02.   In Figure \ref{fig:InOutput} we present the two-dimensional empirical distribution of in and out-degrees of the UTXO transactions. We can make three main observations. First, around $95\%$ of the transactions use not more than $2$ inputs and not more than $3$ outputs. Second, we can observe  a ``horizontal line'' with two outputs and inputs varying from $1$ to $128$. This effect can be explained by a wallet functionality that tries to keep the number of unspent outputs as small as possible. Third, there is a ``vertical'' line using $3$ inputs, and this effect can be explained by using a special output type called ``alias-output'' for the minting of NFTs. These observations show that the actual distribution of inputs and outputs heavily relies on different uses-cases and functionalities. For this reason, we choose to model the distribution of the number of inputs  and outputs as a uniform distribution, the one with the highest entropy, on the typical ``input-output'' relation. Moreover,  the observed conflict rate  is $p_{\textrm{conflict}}\approx 0.03$, and we use the two values $0.01$ and $0.05$ to get bounds on the ``likely'' behaviour. The choices of $p_{\textrm{conflict}} \in \{0.1, 0.5 \}$ are motivated to cover scenarios where an attacker spam conflicts. 
\end{remark}

 \begin{figure}[t]
    \centering
    \includegraphics[width=0.7\textwidth]{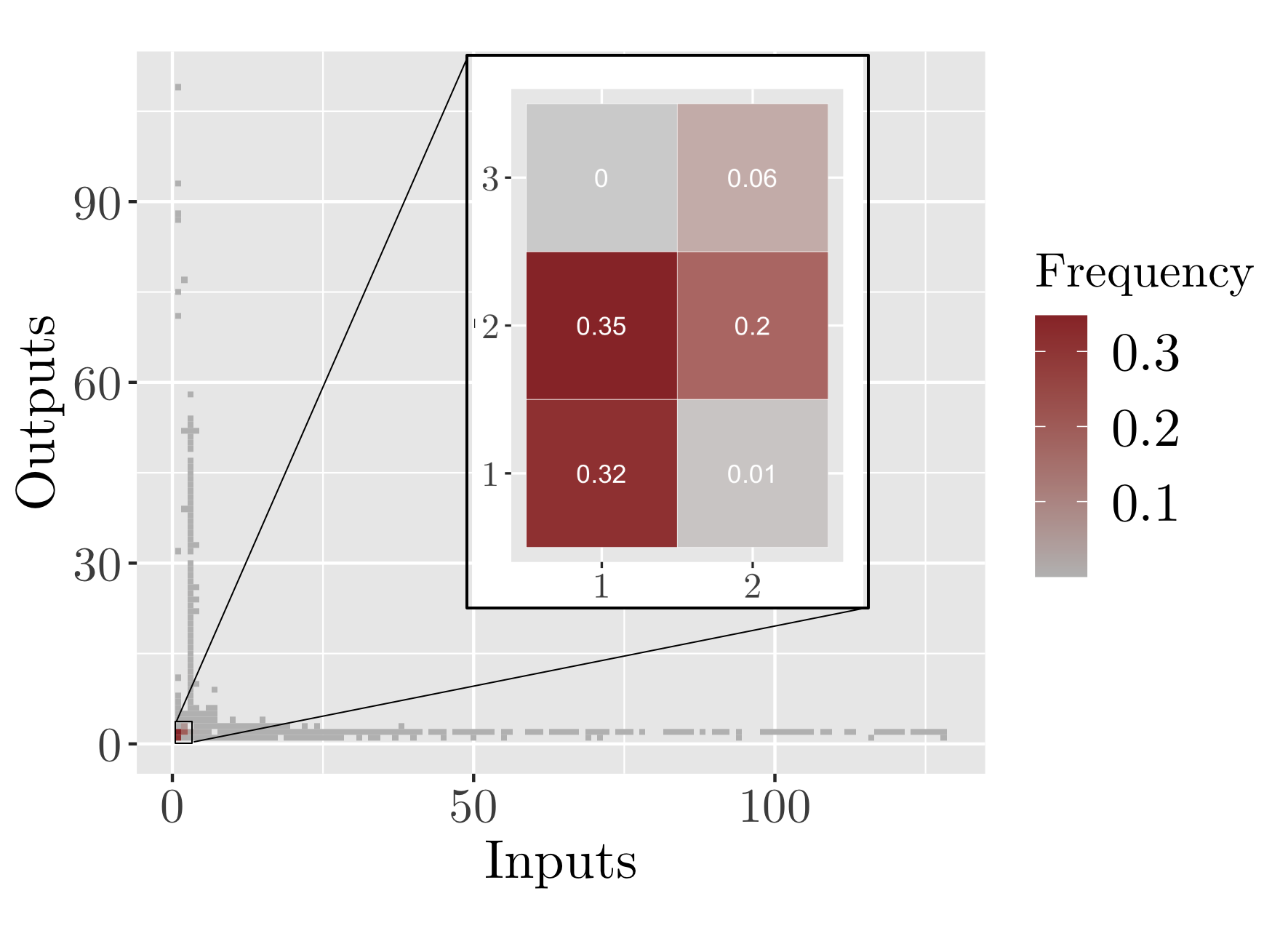}
    \caption{Empirical distribution of the number of inputs and outputs in the Shimmer network (from 2022/09/27 to 2023/01/02). }
    \label{fig:InOutput}
\end{figure}

To store and update the main data structures, we utilize map containers. For each transaction $x$, we store (and update when necessary) the following fields in a reality-based ledger:
\begin{itemize}
    \item $x.\mathrm{InputLabels}$, a list of inputs $\inputs(x)$ from the set of labels $\{0,1\}^{256}$; 
    \item $x.\mathrm{OutputLabels}$, a list of outputs $\outputs(x)$ from the set of labels $\{0,1\}^{256}$; 
    \item $x.\mathrm{Parents}$, a list of parents $\parent{\ledger}{x}$, transactions whose outputs are spent by $x$;
    \item $x\mathrm{.Children}$, a list of children $\child{\ledger}{x}$, transactions that consumes the output(s) of $x$;  
    \item $x.\mathrm{ParentsConflicts}$, a list of the closest conflicts $\conflictsInCone{p}{\ledger}{x}$ in the past cone of $x$; 
    \item $x.\mathrm{ChildrenConflicts}$, a list of the closest conflicts $\conflictsInCone{f}{\ledger}{z}$ in the future cone of $x$;  
    \item $x.\mathrm{DirectConflicts}$, a list of transactions that are directly conflicting with $x$
    \item $x.\mathrm{InputConflictLabels}$, a list of inputs that are consumed with some other transactions (directly conflicting with $x$)
\end{itemize}

Storing branches associated with each transaction and conflict might increase the storage overhead by a factor of the number of conflicts. Thereby, for a transaction $x\in\ledger$, we have implemented a function $\textsc{GetBranch}(\cdot)$ that returns $\lab^{(p)}(x)$, all conflicts in the past cone of $x$. Note that the fields  $\mathrm{.ParentsConflicts}$ and $.\mathrm{ChildrenConflicts}$ allow traversing a DAG on the set of conflicts. Denote this DAG (the DAG on the set of conflicts whose edges are defined by $\mathrm{.ParentsConflicts}$ and $.\mathrm{ChildrenConflicts}$) as $\hat{D}_{\conflictset}$. We note that $\hat{D}_{\conflictset}$ contains all the edges of the Conflict DAG $\ConflictDAG$ and some other extra edges, but the reachability properties of $\ConflictDAG$ are preserved in $\hat{D}_{\conflictset}$, i.e., $\ConflictDAG$ is a transitive reduction of $\hat{D}_{\conflictset}$. 
 
We have implemented a function $\textsc{GetReality}(\cdot)$ that returns the same reality as Algorithms~\ref{alg:selectionBranch}-\ref{alg:selectionConflictGraph} do for the case when the weights of all transactions are zeros. Storing and updating the Branch DAG and the Conflict Graph is impractical for a large number of conflicts since their sizes can be exponential and quadratic in the number of conflicts (see Remark~\ref{rem: complexity of Branch DAG}). Instead, the structure of the DAG $\hat{D}_{\conflictset}$ is utilized in our implementation. Specifically, we iteratively construct the reality (similar to Algorithm~\ref{alg:selectionConflictGraph}): at every step, we choose $c^*$ that is $\ConflictDAG$-maximal and attains the minimal hash (see line~\ref{line: minimal hash among maximal remaining elements}). We have also implemented a function $\textsc{PruneRejectedTransactions}(\cdot)$ that takes as input a reality and prunes all transactions conflicting with at least one conflict from the preferred reality. This function works similarly to Algorithm~\ref{alg: update UTXO DAG after confirmation} when the latter takes as an input the reality obtained from Algorithms~\ref{alg:selectionBranch}-\ref{alg:selectionConflictGraph}.

\begin{figure}[t]
    \centering
    \includegraphics[width=0.9\textwidth]{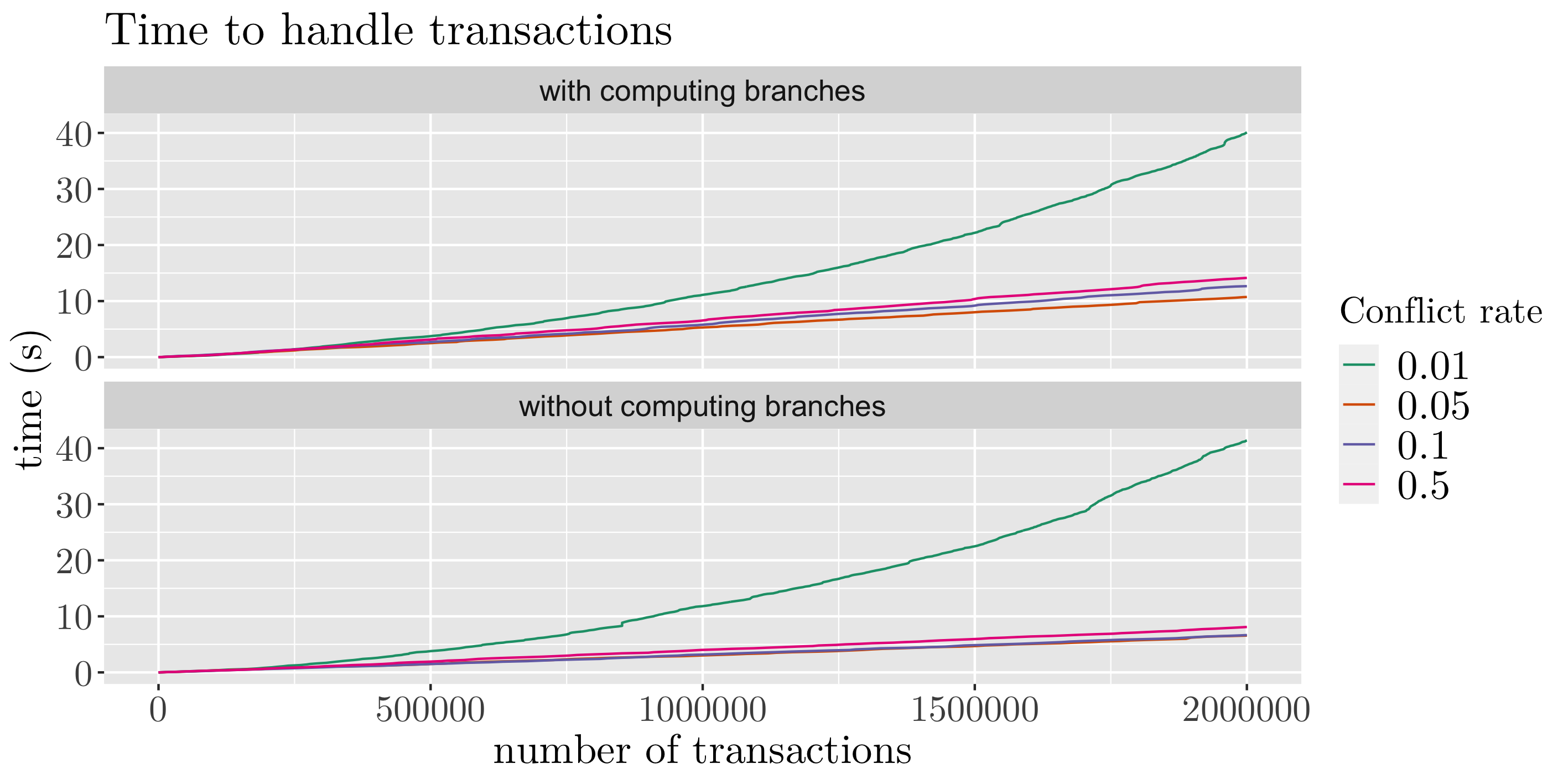}
    \caption{Time to handle the incoming flux of transactions, with and without computing $\textsc{GetBranch}(\cdot)$ for every transaction. The benchmark is taken for different probabilities $p_{\textrm{conflict}}\in\{0.01,0.05,0.1,0.5\}$.}
    \label{fig: growing ledger}
\end{figure}

In Figure~\ref{fig: growing ledger}, for different probabilities $p_{\textrm{conflict}}\in\{0.01, 0.05,0.1,0.5\}$, 
we depict the growth of a reality-based ledger over time with and without computing $\textsc{GetBranch}(\cdot)$ for each new transaction. We employ the model of a randomized stream of transactions described above. The number of transactions is limited by 8GB of RAM allocated to the reality-based ledger.  
The worst performance corresponds to the case when conflicts are located sparsely in the ledger ($p_{\textrm{conflict}}=0.01$). The latter can be explained by the fact that for $p_{\textrm{conflict}}=0.01$, one needs to traverse more vertices in the Ledger DAG on average to update the fields $.\mathrm{ChildrenConflicts}$ and $.\mathrm{ParentsConflicts}$ for transactions in the future and past cones of a new conflicting transaction. Specifically, for each non-conflicting transaction, there could be some paths between the closest conflicts that contain the transaction; one needs to traverse through that transaction in the Ledger DAG as many times as one of such paths is updated.   The computational complexity of $\textsc{GetBranch}(\cdot)$ heavily depends on the height of the closest conflicts of a given transaction in $\ConflictDAG$. For instance, in the considered model, the expected height of a random conflict in $\ConflictDAG$ is asymptotically logarithmic with the total number of conflicts. This is seen in Figure~\ref{fig: growing ledger} as there is no significant performance degradation when one additionally computes the branch for each new transaction by calling $\textsc{GetBranch}(\cdot)$. In all cases, the rate of transactions per second is $50,\!000-150,\!000$ tx/sec.

 \begin{figure}[t]
    \centering
    \includegraphics[width=0.9\textwidth]{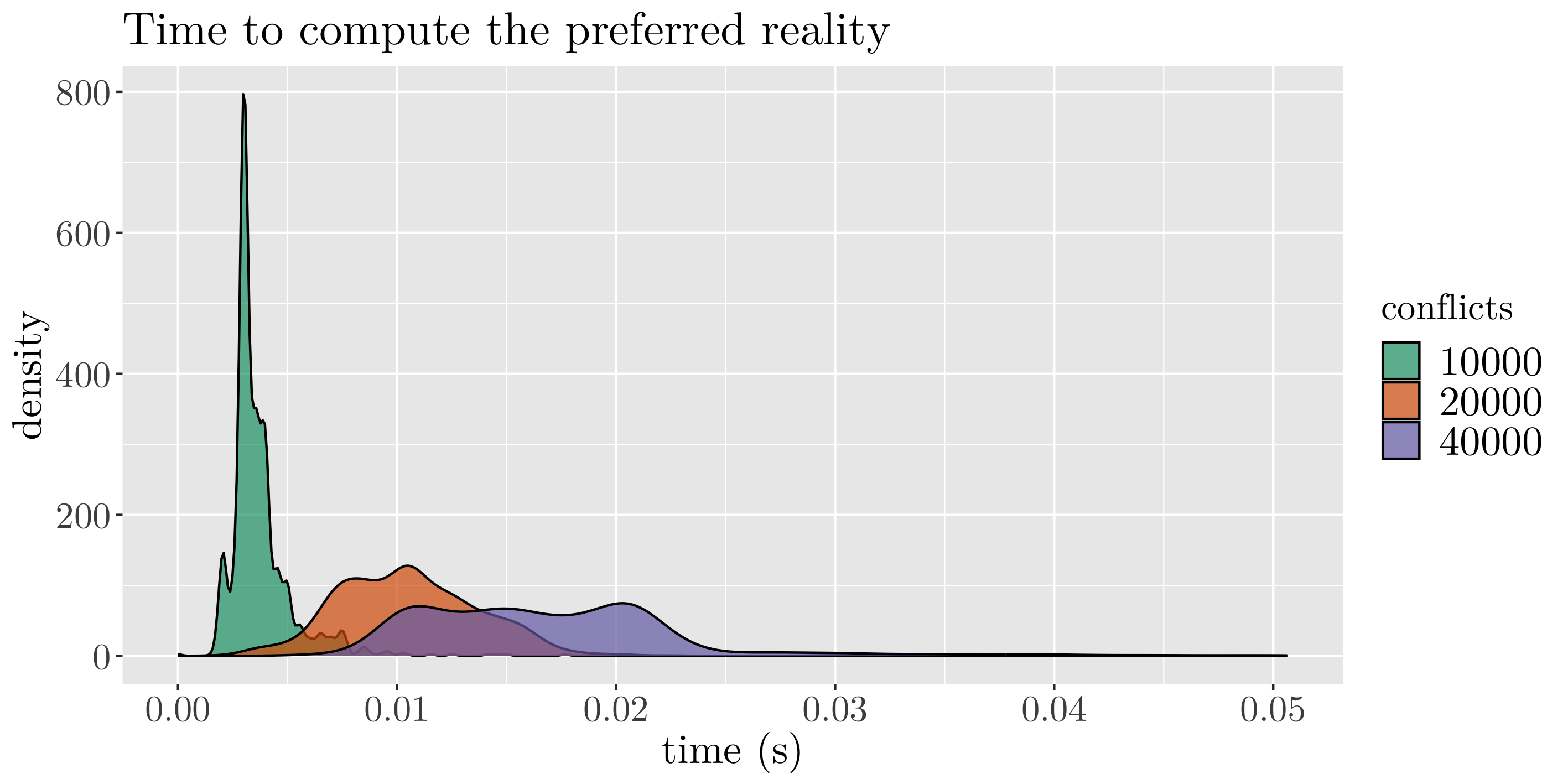}
    \caption{Time to compute the preferred reality for different numbers of conflicts.}
    \label{fig: time preferred reality}
\end{figure}

In Figure~\ref{fig: time preferred reality}, we show the statistics of the time that it takes to compute $\textsc{GetReality}(\cdot)$ when the number of conflicts $|\conflictset|\in \{10000,20000, 40000\}$.  To generate conflicts, we utilize the same model of randomized transactions with $p_{\textrm{conflict}}=0.05$. Recall that, by Proposition~\ref{prop: complexity of algorithm}, the complexity of the reality selection algorithm in the worst case can be bounded as $O(|\conflictset|^2)$ if one follows Algorithm~\ref{alg:selectionBranch}-\ref{alg:selectionConflictGraph}. In a more practical implementation that we employ in our simulation, the worst-case complexity is even worse as finding $\ConflictDAG$-maximal elements and the element with the largest hash among the maximal elements in a list are not necessarily $O(1)$-operations. However,  in Figure~\ref{fig: time preferred reality}, the amortized complexity of our benchmarks seems linear with the number of conflicts $|\conflictset|$. 
We can observe a concentrated distribution for $10,\!000$ conflicts and that the distribution flattens as the number of conflicts increases. Interestingly, the eventual size of the preferred reality is more robust in the increase of conflicts; see Figure~\ref{fig: size preferred reality}.

 \begin{figure}[t]
    \centering
    \includegraphics[width=0.9\textwidth]{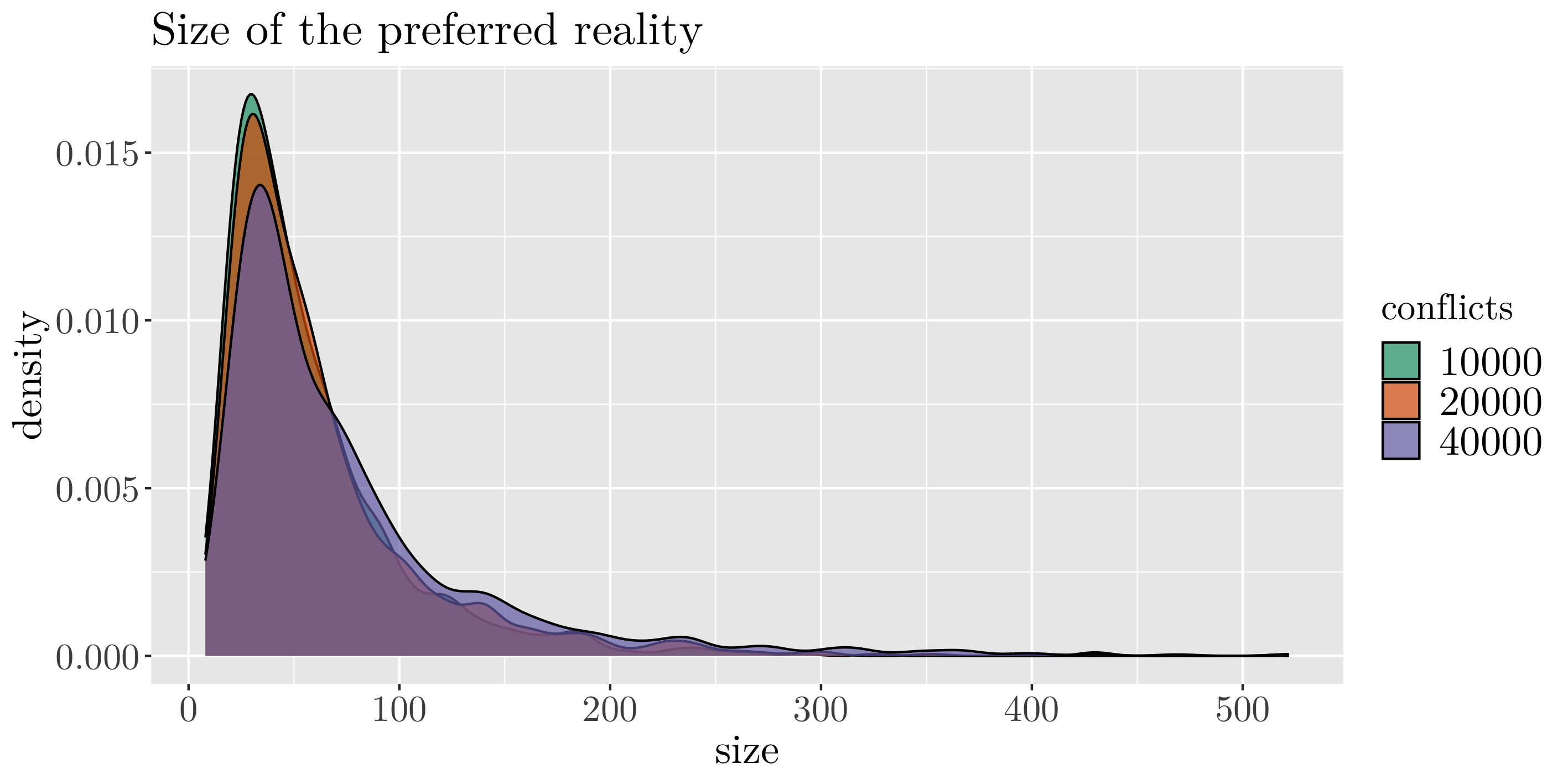}
    \caption{Size of the preferred reality for different numbers of conflicts.}
    \label{fig: size preferred reality}
\end{figure}

As the time to calculate the preferred reality increases with the number of conflicts, we periodically remove confirmed transactions from our data structures. In Figure~\ref{fig: ledger growth with pruning}, we depict the growth of a reality-based ledger, the number of confirmed transactions and conflicts over time. We generate a stream of transactions with parameter $p_{\textrm{conflict}}=0.01$ (the case providing the worst performance in Figure~\ref{fig: growing ledger}). Whenever the number of conflicts exceeds the set upper limit of $5,\!000$, we apply the functions $\textsc{GetReality}(\cdot)$ and $\textsc{PruneRejectedTransactions}(\cdot)$. 
In addition, we treat all remaining transactions as confirmed and update the data structure of the reality-based ledger by taking the confirmed ledger state as a new genesis. 
We highlight that the number of pruned transactions is much larger than that of pruned conflicts since a transaction that indirectly spends from an output of a pruned conflict must also be pruned. In summary, by periodically pruning rejected transactions and snapshotting confirmed transactions, the protocol allows handling an ongoing stream of transactions without performance loss. The rate of confirmed transactions per second in this simulation is over $70,\!000$ txs/sec. This experiment indicates the robustness and scalability of our proposed reality-based UTXO ledger for real-world use cases.

 \begin{figure}[t]
    \centering
    \includegraphics[width=0.9\textwidth]{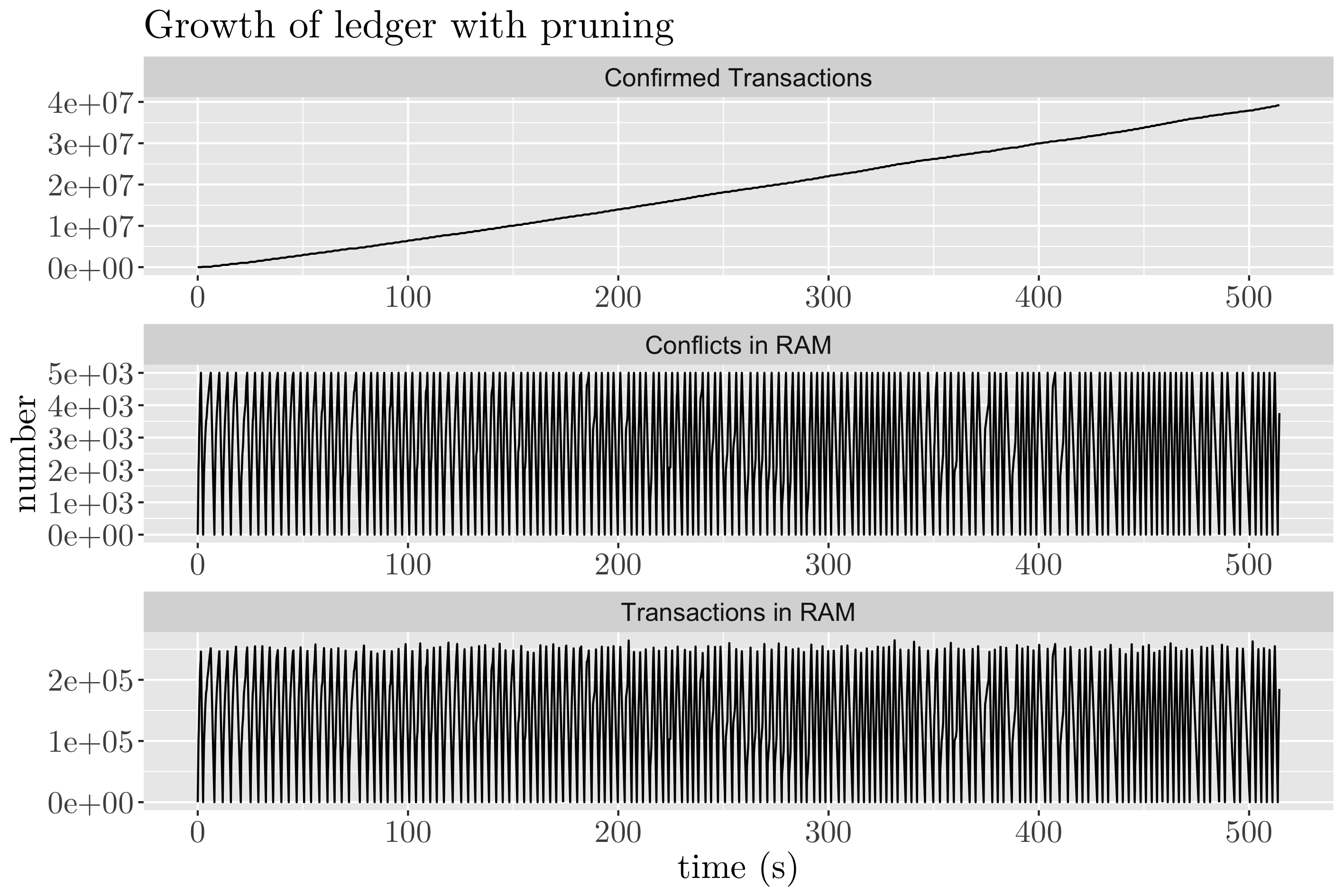}
    \caption{Growth of the ledger and number of transactions and conflicts kept in RAM.}
    \label{fig: ledger growth with pruning}
\end{figure}

\section{Future Work}
The Reality-based Ledger provides a framework for parallel transaction processing capability. Typically in blockchains,  transactions are processed in blocks or batches, creating a total order. This linearisation creates an artificial bottleneck in the propose and vote paradigm of DLTs. However, particularly in a UTXO setting, this is not necessary, and blockchain systems can be designed using the presented framework. 
We follow this approach in~\cite{OTV}, which builds on the foundations laid out in this paper.

We provided a first quantitative analysis of our proposed algorithms. Future work in this area should focus on benchmarking the performance of different DAG-based DLTs against each other. This could involve comparing throughput by measuring different kinds of transactions, e.g.~simple value transfers or transactions of smart contracts, as well as assessing the potential for quick view changes. Additionally, it is essential to consider the manipulation potential in the various systems, including the possibility of annulling or frontrunning transactions. This future work must also address the challenge of constructing appropriate and meaningful performance measures and test scenarios for the different DLT solutions, as their distinct natures and other underlying use cases will impact the results. These benchmarks will ensure a fair and more comprehensive evaluation of the strengths and limitations of various DAG-based DLT solutions.

\subsubsection*{Acknowledgments}
The authors would like to 
thank precious staff members of the IOTA Foundation and members of the IOTA community for their feedback and criticism.

\bibliographystyle{plain}

\bibliography{bibliography}
\end{document}